\documentclass{article}

\PassOptionsToPackage{numbers, compress, sort}{natbib}
\usepackage[preprint]{neurips_2025}
\usepackage{lineno}

\usepackage[utf8]{inputenc}
\usepackage[british]{babel}  
\usepackage[autostyle]{csquotes} 
\usepackage{graphicx} 
\usepackage{amsfonts, amsmath, amssymb, amsthm}
\usepackage[ruled, noend]{algorithm2e}
\usepackage[noend]{algorithmic}
\usepackage{hyperref}
\usepackage{mathtools}
\usepackage{booktabs}
\usepackage[T1]{fontenc}
\usepackage{nicefrac}
\usepackage{microtype}
\usepackage{cleveref}
\usepackage{xcolor}
\usepackage{caption}
\usepackage{subcaption}
\usepackage{enumitem}

\usepackage{color-edits}
\addauthor{dn}{blue}
\addauthor{vp}{olive}
\addauthor{nd}{red}

\newcommand{\declarecolor}[2]{\definecolor{#1}{RGB}{#2}\expandafter\newcommand\csname #1\endcsname[1]{\textcolor{#1}{##1}}}

\declarecolor{White}{255, 255, 255}
\declarecolor{Black}{0, 0, 0}
\declarecolor{Maroon}{128, 0, 0}
\declarecolor{Coral}{255, 127. 80}
\declarecolor{Red}{182, 21, 21}
\declarecolor{LimeGreen}{50, 205, 50}
\declarecolor{DarkGreen}{0, 90, 0}
\declarecolor{Navy}{0, 0, 128}

\definecolor{plotblue}{HTML}{377eb8}
\definecolor{plotorange}{HTML}{ff7f00}
\definecolor{plotgreen}{HTML}{4daf4a}
\definecolor{darkblue}{rgb}{0, 0, 0.5}
\hypersetup{colorlinks=true, citecolor=darkblue, linkcolor=darkblue, urlcolor=darkblue}

\newtheorem{theorem}{Theorem}[section]
\newtheorem{lemma}[theorem]{Lemma}
\newtheorem{fact}[theorem]{Fact}

\newtheoremstyle{named}{}{}{\itshape}{}{\bfseries}{.}{.5em}{\thmnote{#3 }#1}
\theoremstyle{named}

\makeatletter
\def\namedlabel#1#2{\begingroup
   \def\@currentlabel{#2}%
   \label{#1}\endgroup
}
\makeatother

\newcommand{\cV}{\mathcal{V}}

\newcommand{\ie}{{i.e.,~\xspace}}
\newcommand{\eg}{{e.g.,~\xspace}}

\usepackage{accents}

\newcommand{\LB}{\mathrm{LB}}
\newcommand{\UB}{\mathrm{UB}}

\usepackage[utf8]{inputenc} 
\usepackage[T1]{fontenc}    
\usepackage{hyperref}       
\usepackage{url}            
\usepackage{booktabs}       
\usepackage{amsfonts}       
\usepackage{nicefrac}       
\usepackage{microtype}      
\usepackage{xcolor}         

\usepackage{csquotes}

\title{Breaking Distortion-free Watermarks in Large Language Models}

\author{
Shayleen Reynolds\thanks{J.P.Morgan AI Research, $\{$\texttt{shayleen.reynolds, tuandung.ngo, saheed.obitayo, niccolo.dalmasso, vamsi.k.potluru, manuela.veloso}$\}$\texttt{@jpmchase.com}}$\,$
    \And
   Hengzhi He \thanks{Department of Statistics and Data Science, UCLA, $\{$\texttt{hengzhihe, guangcheng}$\}$\texttt{@ucla.edu}}$\,$
  \And
  Dung Daniel T. Ngo\footnotemark[1]
  \And
  Saheed Obitayo\footnotemark[1]
  \And
  Niccol\`o Dalmasso\footnotemark[1]
  \And
  Guang Cheng \footnotemark[2]
  \And
  Vamsi K. Potluru\footnotemark[1]
  \And
  Manuela Veloso\footnotemark[1]
}

\begin{document}

\maketitle

\begin{abstract}
  In recent years, LLM watermarking has emerged as an attractive safeguard against AI-generated content, with promising applications in many real-world domains. However, there are growing concerns that the current LLM watermarking schemes are vulnerable to expert adversaries wishing to reverse-engineer the watermarking mechanisms. Prior work in \textquote{breaking} or \textquote{stealing} LLM watermarks mainly focuses on the distribution-modifying algorithm of \citet{kirchenbauer2023watermark}, which perturbs the logit vector before sampling. In this work, we focus on reverse-engineering the other prominent LLM watermarking scheme, distortion-free watermarking \citep{kuditipudi2024robust}, which preserves the underlying token distribution by using a hidden watermarking key sequence. We demonstrate that, even under a more sophisticated watermarking scheme, it is possible to \emph{compromise} the LLM and carry out a \emph{spoofing} attack, \ie generate a large number of (potentially harmful) texts that can be attributed to the original watermarked LLM. Specifically, we propose using adaptive prompting and a sorting-based algorithm to accurately recover the underlying secret key for watermarking the LLM. Our empirical findings on \textrm{LLAMA-3.1-8B-Instruct}, \textrm{Mistral-7B-Instruct}, \textrm{Gemma-7b}, and \textrm{OPT-125M} challenge the current theoretical claims on the robustness and usability of the distortion-free watermarking techniques. 
\end{abstract}

\section{Introduction}
\label{sec:intro}
Recent advances in generative models have significantly improved their capabilities and applicability across various real-world domains. Notably, models like ChatGPT~\citep{chatgpt} and other large language models (LLMs) can now generate text closely resembling human-written content. However, as both businesses and individuals have rapidly adopted generative models, there is a growing concern within the research community about their potential for malicious use. To address this issue, a growing body of research around \emph{watermarking} LLM-generated text has emerged~\citep{kirchenbauer2023watermark, kuditipudi2024robust, aaronson2023openai, piet2024markwordsanalyzingevaluating, zhang2024remark, ning2024mcgmark}. The primary strategy in this line of research involves embedding a \emph{hidden} signal (\ie a secret watermark key) within the generated text, which any third party with knowledge of the key can reliably detect.

While these watermarking techniques offer reliable and robust statistical guarantees to verify LLM-generated texts, they still fall short in addressing the potential attack models posed by malicious actors \citep{jovanovic2024watermarkstealinglargelanguage, zhang2024largelanguagemodelwatermark, pang2024no, wu2024bypassingllmwatermarkscoloraware, gloaguen2024black, gloaguen2024discovering}. Previous research on LLM watermarking often focuses on robustness against common attacks, such as deletion, insertion, and substitution, to simulate the behavior of users attempting to evade content detectors. For instance, a student might slightly modify a machine-generated essay by altering a few sentences to avoid detection by their professor. However, a determined adversary could go further by reverse-engineering the watermarking scheme. By repeatedly querying the API of the watermarked LLM, they could \emph{steal} the watermark by approximating the hidden secret key. Once estimated, the most significant threat is \emph{spoofing}, where an attacker generates (potentially harmful) text that appears to be watermarked. Being able to generate large volumes of \emph{spoofed} content with minimal computational effort not only undermines the intended purpose of a watermark but is also a reputational risk for the LLM providers, whose model could have been falsely attributed to harmful or incorrect content.

Prior work on \emph{watermark stealing} primarily studies the distribution-modifying algorithm by \citet{kirchenbauer2023watermark} and its variants. In contrast, our focus is on the other prominent watermarking scheme by \citet{kuditipudi2024robust}, which is \emph{distortion-free}, \ie the watermark does not change the underlying token distribution. A significant difference between the two watermarking techniques is that \citet{kuditipudi2024robust} uses a randomized watermark key, creating a correlation between the LLM-generated text and this secret key. During detection, a third party with this secret watermark key can efficiently check for the correlation and verify whether the text is watermarked. Furthermore, prior work that attempted to break this distortion-free watermark specifically focus on the exponential minimum sampling variant (EXP). In this work, we complete the line of work on stealing \citet{kuditipudi2024robust}'s watermark by investigating the inverse transform sampling (ITS) variant (\Cref{sec:protocol}). 

Considering this approach, we propose a sorting-based algorithm to accurately estimate the secret watermark key and enable \emph{spoofing} attacks with only a few samples from the watermarked LLM. Our spoofed outputs both (i) pass the original watermark detection test, and (ii) maintain similar text quality (measured with perplexity score and cosine similarity) compared to the non-watermarked and watermarked outputs from the base LLM. Notably, we observe successful spoofing attempts even when the attacker has a partial knowledge of the secret key and permutation used for watermarking. 
\begin{figure}
     \centering
     \begin{subfigure}[b]{\textwidth}
        \centering
        \includegraphics[width=\linewidth]{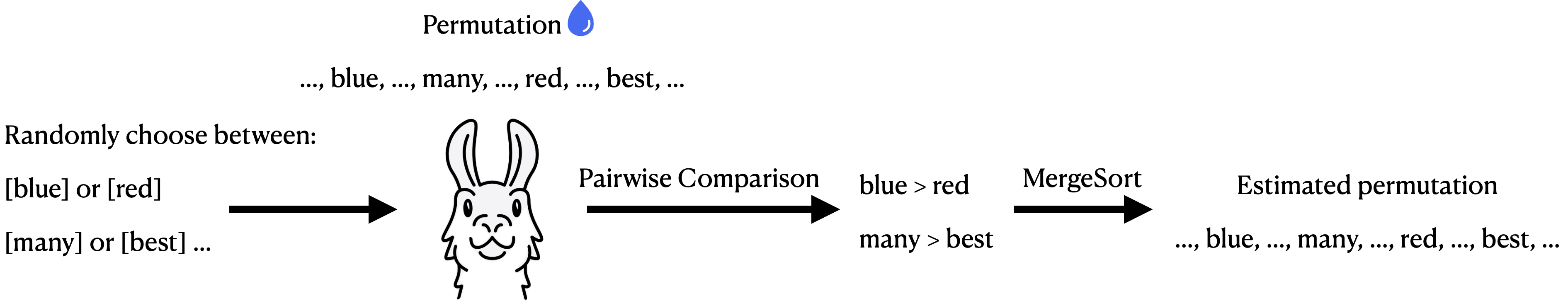}
         \label{fig:overview_1}
     \end{subfigure}
     \vfill
     \begin{subfigure}[b]{\textwidth}
         \centering
         \includegraphics[width=\linewidth]{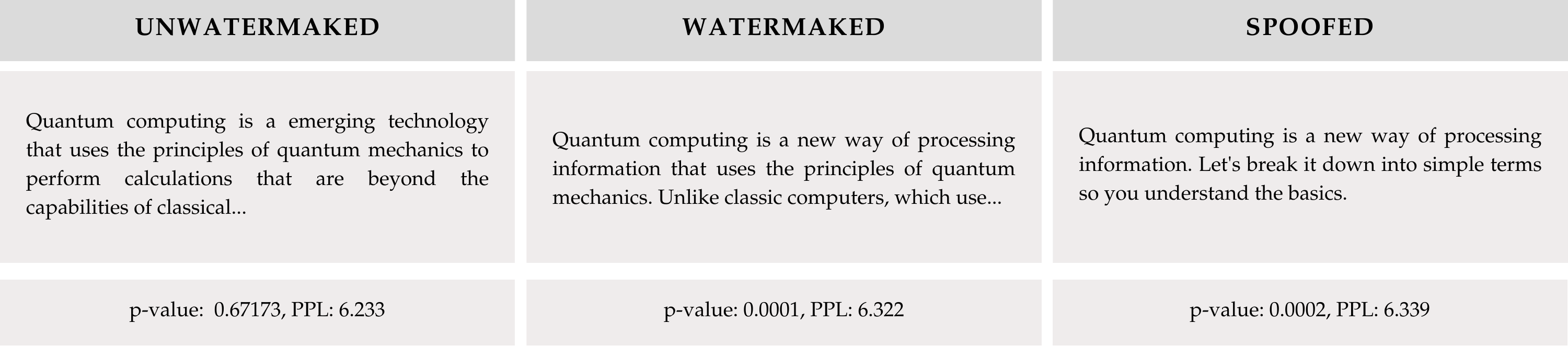}
         \label{fig:spoofed_table}
     \end{subfigure}
     \caption{Overview of our watermark stealing framework. \textbf{Top}: Illustrative example of~\Cref{alg:reverse_perm} to estimate the secret permutation over the LLM's vocabulary. \textbf{Bottom}: Outputs using LLAMA-3.1-8B-Instruct \textbf{(left)}: without watermark, \textbf{(middle)}: with watermark, \textbf{(right)}: spoofing attack. The $p$-value reflects the watermark detection confidence, and the perplexity (PPL) measures text quality.}
\end{figure}
\paragraph{Overview of Paper.} We introduce a watermark stealing framework for the distortion-free LLM watermarking technique from \citet{kuditipudi2024robust}. Specifically, in this work we focus on the watermark approach using inverse transform sampling. At a high level, the decoder used in inverse transform sampling has two main components: a secret key sequence of uniformly random variables and a permutation over the entire vocabulary. With a sequence of well-crafted prompts, our algorithm can probe the underlying watermark to accurately learn the permutation over the vocabulary used to generate watermarked text. Once the permutation is recovered, we repeatedly query the watermarked LLM and obtain an accurate estimate of the secret key sequence by iteratively narrowing down the confidence interval with each observed sample. After both components of the watermark decoder have been learned, we can successfully perform \emph{spoofing} attacks on the watermarked LLM, \ie generate texts that appear to be watermarked.
Overall, we make the following contributions:
\begin{itemize}
    \item We provide a framework that accurately estimates underlying parameters of the distortion-free watermarking algorithm \citep{kuditipudi2024robust} under different threat models. Specifically, we target the watermark algorithm based on inverse transform sampling (ITS), which uses a sequence of random secret keys and a permutation over the vocabulary to watermark the LLM. 
    \item With this secret watermark key estimation, we demonstrate a simple spoofing attack with a high success rate on four LLMs (LLAMA-3.1-8B-Instruct~\citep{grattafiori2024llama3herdmodels}, Mistral-7B-Instruct~\citep{jiang2023mistral7b}, Gemma-7b~\citep{gemmateam2024gemmaopenmodelsbased}, and OPT-125M~\citep{zhang2022opt}). In conjunction with prior work in watermark stealing, these results highlight the need for refining state-of-the-art watermarking techniques for language models.
\end{itemize}
\section{Related Work}
\label{sec:related_work}
\paragraph{LLM Watermark.} Advances in large language models have led to promising applications in various real-world domains. However, there is growing concern that language models may be misused for spreading fake news, generating harmful content, and facilitating academic dishonesty. In response, a growing body of research has proposed watermarking as a framework to reliably detect LLM-generated content and mitigate the potential for malicious use \citep{kirchenbauer2023watermark, kuditipudi2024robust, aaronson2023openai, piet2024markwordsanalyzingevaluating, zhang2024largelanguagemodelwatermark, ning2024mcgmark}. These approaches embed an `invisible' watermark in the model-generated output, which can later be identified and verified using a secret key. Existing LLM watermarking schemes share various desirable properties: (i) the watermark should be easily detected given knowledge of the secret key; (ii) the watermark should not degrade model-generated outputs; (iii) the watermark should be robust against adversarial attacks; and (iv) the watermark should not be easily stolen for spoofing or removal attacks. However, recent work on `watermark stealing' directly challenges these theoretical claims.
\paragraph{LLM Watermark Stealing.} The literature on LLM watermark and watermark stealing primarily considers two attacks: \emph{scrubbing} and \emph{spoofing}. Prior work on watermark scrubbing studies effective watermark removal by paraphrasing \citep{krishna2023paraphrasing, zhang2024watermarkssandimpossibilitystrong,jovanovic2024watermarkstealinglargelanguage,cheng2025revealing,krishna2023paraphrasing} or leveraging the LM's side information \citep{pang2024no}. 

In this work, we instead focus on the `spoofing' attack. The primary approach in this literature is first to estimate the watermarking scheme, then embed this secret key approximation into arbitrary content to generate spoofing attacks. The first work that comprehensively studied spoofing is  \citet{sadasivan2025aigeneratedtextreliablydetected}, which targets \citet{kirchenbauer2023watermark}'s distribution-modifying watermark. The authors query the watermarked LLM a million times to learn the underlying token pair distribution, then manually compose texts to spoof the watermark. Follow-up works~\citep{jovanovic2024watermarkstealinglargelanguage, gu2024learnabilitywatermarkslanguagemodels, liu2024survey} further highlight the importance of spoofing by operating in more realistic settings. Notably, \citet{jovanovic2024watermarkstealinglargelanguage} spoofs ~\citet{kirchenbauer2023watermark}'s soft watermark without access to the watermark $z$-score detector and no base (non-watermarked) responses. The work most related to ours is  \citet{pang2024no}, who provide spoofing attacks on the three most prominent LLM watermarking schemes: KGW~\citep{kirchenbauer2023watermark}, Unigram~\citep{zhao2023provablerobustwatermarkingaigenerated}, and EXP~\citep{kuditipudi2024robust}. Instead of directly estimating the watermarking scheme, the authors leverage the public detection API to enable more sample-efficient spoofing and propose a differential privacy approach to mitigate the risks of having a public detection API. Our work differs from this work in that we consider the inverse transform sampling (ITS) watermark approach in~\citet{kuditipudi2024robust} instead of the exponential minimum sampling (EXP) approach. \citet{kuditipudi2024robust} suggests the EXP watermark in practice when robustness is of higher priority and the ITS watermark in practice when detection throughput is of higher priority. Thus, our results complement their findings and complete the spoofing attacks on both watermark approaches in \citet{kuditipudi2024robust}.   
\section{Preliminary}
\label{sec:prelim}
In this section, we provide relevant background on the inverse-transform mapping-based watermarking method proposed in \citet{kuditipudi2024robust}, our notations, and a concrete set of threat models. In the following, we use $x$ to denote a sequence of tokens, $x_i \in \cV$ is the $i$-th token in the sequence, and $\cV$ is the vocabulary. For $n \in \mathbb{N}^+$, we write $[n]$ to denote the set $\{1, \cdots, n\}$.  
\subsection{Watermark Generation and Interaction Protocol}
\label{sec:protocol}
Let $p: \cV^* \rightarrow \Delta(\cV)$ be an auto-regressive \emph{language model} (LM) that maps a string of arbitrary length to a probability distribution over the vocabulary $\Delta(\cV)$. Given a prefix $x \in \cV^*$, we denote the conditional probability distribution over the next token by $p(\cdot | x)$. Let $\Xi$ represent the space of watermark key elements, and let $\xi \in \Xi^n$ be a secret key sequence. For simplicity, we assume that each element $\xi_{i \in [n]}$ of $\xi$ belongs to the interval $[0,1]$. Let $\pi: [|\cV|] \rightarrow \cV$ be a fixed bijective permutation mapping from integers to tokens, where $\pi(i)$ represents the $i$-th token in the vocabulary sorted according to the permutation $\pi$. The watermarking mechanism is based on inverse transform sampling (ITS) with the following interaction protocol:
\begin{enumerate}
    \item \textbf{Key and Mechanism Sharing}: The LM provider shares a secret key sequence $\xi = (\xi_1, \xi_2, \dots, \xi_n) \in \Xi^n$ and a bijective permutation function $\pi$ with the detector.
    
    \item \textbf{User Prompt}: The user provides a prompt $x \in \cV^*$ to the LM provider.
    
    \item \textbf{Watermarked Text Generation}: The LM provider generates watermarked text $Y \in \cV^*$ by applying the following procedure at each token generation step:
        \begin{enumerate}
            \item Compute the CDF of the LM’s probability distribution over the vocabulary: $C_k = \sum_{i=1}^k p(\pi(i) | x), \quad \forall k \in [|\cV|]$. 
            \item Retrieve the corresponding key element for the $t$-th token generation: $\zeta = \xi_{t \mod
            n}$.
            \item Identify the smallest index $k$ such that $C_k \geq \zeta$.
            \item Select the token $\pi(k)$ as the next token in $Y$.
            \item Append the selected token to the current prefix and update it: $x \gets x \oplus \pi(k)$.
            \item Repeat until the desired text length is reached or another termination condition is met.
        \end{enumerate}
\end{enumerate}

The details of the generation process are summarized in Algorithm~\ref{alg:generate} (\Cref{appendix:prelim}).

\subsection{Watermark Detection}
We adopt a methodology similar to that proposed in \citet{kuditipudi2024robust} for watermark detection, employing a permutation test to compute the $p$-value. If the test returns a small $p$-value, then the text is likely to be watermarked; otherwise, if the $p$-value is large, then the text is likely not watermarked. The details are summarized in Algorithm~\ref{alg:detection} in Appendix~\ref{appendix:prelim}.

\subsection{Attacker Model}
\paragraph{Attacker's Objective and Motivation.} Our work focuses on the `spoofing' attack, which aims to generate harmful or incorrect outputs that carry the original LM provider's watermark. For example, an attacker can `spoof' the watermarked LM, generate fake news or defamatory remarks, and post them on social media. The attacker can damage the LM provider and their model's reputation by claiming that the innocent LM provider's model generates these harmful texts.
\paragraph{Attacker's Capabilities.} We make the following assumptions on the attacker's capabilities, following prior work on watermark stealing~\citep{jovanovic2024watermarkstealinglargelanguage, pang2024no}. The attacker has black-box access to the complete generation of the watermarked LM and is aware of the watermark employed by the LM provider. The attacker aims to use a small number of queries to the watermarked LM to build an estimation of the underlying watermarking scheme parameterized by the secret key $\xi$ and permutation $\pi$. We further assume that the detection API is available to the public. The attacker can query the detection API and obtain the watermark confidence score in terms of $p$-value. This assumption allows the attacker to verify the effectiveness of their `spoofing' attack. Finally, as is common practice in prior work~\citep{Naseh_2023, ouyang2022traininglanguagemodelsfollow, pang2024no} and enabled by OpenAI's API, we assume that the top tokens at each position and their probabilities are returned to the attacker.   

\section{Methodology: Breaking Distortion-free Watermarks}
\label{sec:methodology}
In this section, we explore three different regimes for breaking the distortion-free watermark, focusing on secret key estimation and recovering the token permutation used in the watermarking process.

\paragraph{Threat Models.}
The goal of the attack is to reverse-engineer both the permutation \(\pi\) and the secret keys \(\{\xi_i\}\) by querying a LM with a variety of carefully designed prompts. We assume that (i) the model consistently employs the same secret key sequence and permutation for text generation and (ii) the adversary can interact with the model through these crafted prompt queries to extract information about its watermarking process \footnote{Note that the original watermark paper~\citep{kuditipudi2024robust} also suggests using a single random permutation in practice to reduce overhead in both watermark generation and detection.}. Concretely, we consider the following three threat models:
\begin{enumerate}
    \item \textbf{Known \(\{\xi_i\}\), Estimate \(\pi\):} In this model, the attacker knows the secret key sequence \(\{\xi_i\}\) used to generate each sample, but the permutation \(\pi\) governing the token generation is unknown. The attacker aims to recover \(\pi\) based on the observed watermarked outputs and the known $\{\xi_i\}$.
    %
    \item \textbf{Known \(\pi\), Estimate \(\{\xi_i\}\):} In this model, the permutation \(\pi\) is publicly available or otherwise known to the attacker. The attacker's goal is then to infer the secret key sequence \(\{\xi_i\}\) from the corresponding watermarked samples and the known \(\pi\). 
    \item \textbf{Both \(\{\xi_i\}\) and \(\pi\) are Unknown:} The most realistic and challenging model where the attacker has no knowledge of the secret key sequence \(\{\xi_i\}\) and the permutation \(\pi\). Under this model, the attacker must simultaneously recover both parameters solely from the observed watermarked outputs.
\end{enumerate}
\subsection{{Known secret key sequence \(\{\xi_i\}\), reverse-engineer permutation \(\pi\)}}
When the random key sequence \(\{\xi_i\}\) is known, the attacker's goal is to reconstruct the permutation map \(\pi\). In this regime, we assume that the attacker can modify the random key sequence \(\{\xi_i\}\) to perform queries for inferring the permutation. To facilitate this, we construct a dataset of (secret key, prompt) pairs \(\{ (\xi_i, x_i) \}\) as follows (details are summarized in \Cref{alg:dataset_construction}):

\begin{enumerate}
    \item \textbf{Fixed Prompt:} 
    A fixed prompt \(x_{\text{prefix}}\), \eg `Once upon a time,' is used across multiple queries to ensure the conditional probability distribution \(p(\cdot | x_{\text{prefix}})\) remains the same.
    
    \item \textbf{Multiple Queries:}
    For each query to the LM, a different random key \(\xi_i \in [0, 1]\) is used. The key \(\xi_i\) determines the CDF threshold for selecting the first token \(x_i\) in the response. 
    
    \item \textbf{Dataset Construction:}
    By repeatedly querying the language model with the fixed prompt \(x_{\text{prefix}}\) and different random keys \(\xi_i\), we obtain a collection of pairs \(\{ (\xi_i, x_i) \}\), where each \(x_i\) is the first token selected in response to \(\xi_i\).
\end{enumerate}

%
\begin{algorithm}[!t]
\caption{Dataset Construction for Reverse-Engineering the Permutation}
\label{alg:dataset_construction}
\begin{algorithmic}[1]
\renewcommand{\algorithmicrequire}{\textbf{Input:}}
\renewcommand{\algorithmicensure}{\textbf{Output:}}
\REQUIRE Fixed prompt \(x_{\text{prefix}}\), number of queries \(N\)
\ENSURE Dataset \(\mathcal{D} = \{ (\xi_i, x_i) \mid i \in [N] \} \)

\FOR{\(i = 1\) to \(N\)}
    \STATE Sample a random key \(\xi_i \sim \mathcal{U}[0,1]\).
    \STATE Query the language model with the fixed prompt \(x_{\text{prefix}}\) using \(\xi_i\).
    \STATE Receive response and extract the first token \(x_i\).
\ENDFOR
\RETURN \(\mathcal{D}\)
\end{algorithmic}
\end{algorithm}

Then, we utilize a sorting-based algorithm to reverse engineer by first sorting the observed data pairs $\{\xi_i, x_i \}$ in ascending order of $\xi_i$. Then, we record the order in which each unique token $x_i$ appears for the first time. This sequence of first occurrences defines the estimated permutation $\pi$ over the vocabulary. The detail of this sorting-based algorithm is summarized in Algorithm \ref{alg:reverse_perm}.
%
    
%

\begin{algorithm}[H]
\caption{Reverse-Engineer Secret Permutation $\pi$}
\label{alg:reverse_perm}
\begin{algorithmic}[1]
\renewcommand{\algorithmicrequire}{\textbf{Input:}}
\renewcommand{\algorithmicensure}{\textbf{Output:}}
\REQUIRE Dataset $\mathcal{D} = \{ (\xi_i, x_i) \mid i \in [N] \}$
\ENSURE Recovered permutation $\pi$
\STATE Sort $\mathcal{D}$ in ascending order by $\xi_i$ to obtain $\{(\xi_{(1)}, x_{(1)}), (\xi_{(2)}, x_{(2)}), \ldots, (\xi_{(N)}, x_{(N)})\}$.
\STATE Initialize empty list $R = \varnothing$.
\FOR{$i = 1$ to $N$}
    \IF{$R$ is empty or $x_{(i)}$ is not the last element of $R$}
        \STATE Append $x_{(i)}$ to $R$.
    \ENDIF
\ENDFOR
\STATE Set $\pi \gets R$.
\RETURN $\pi$.
\end{algorithmic}
\end{algorithm}

\subsection{Known permutation \(\pi\), reverse-engineer secret key sequence \(\{\xi_i\}\)}
\label{sec:reverse_engineer_key}

In this regime, where the permutation \(\pi\) is known to the adversary, each generated token \(x\) can be mapped to its corresponding index in the vocabulary via the inverse mapping $k = \pi^{-1}(x)$.
Given an index \(k\), the secret key \(\xi\) corresponding to the token $x$'s selection must lie within the following interval determined by the cumulative distribution function of the language model, \ie $C_{k-1} < \xi \leq C_k,$
where the $C_{k-1}$ and $C_k$ are the CDF values with respect to permutation $\pi$ defined as
%
$C_k = \sum_{j=1}^{k} p\bigl(\pi(j) \mid \text{prefix}\bigr)$
with the convention \(C_0 = 0\).

\begin{algorithm}[!t]
\caption{Reverse-Engineering Pseudorandom Secret Key from Watermarked Outputs}
\label{alg:reverse_prng}
\begin{algorithmic}[1]
\renewcommand{\algorithmicrequire}{\textbf{Input:}}
\renewcommand{\algorithmicensure}{\textbf{Output:}}
\REQUIRE 
\begin{itemize}
    \item A set of output sentences \(\mathcal{Y}\) generated by the watermarked model.
    \item Known permutation \(\pi\) (and its reverse mapping \(\pi^{-1}\)).
    \item Watermark key length \(n\) (\ie there are \(n\) pseudorandom numbers \(\xi_1, \xi_2, \dots, \xi_n\)).
\end{itemize}
\ENSURE Estimated lower bound $\LB_i$ and upper bound $\UB_i$ for each pseudorandom number \(\xi_{i: i \in [n]}\).

\STATE \textbf{Initialization:} For \(i = 1\) to \(n\), set \(\LB_i \gets 0\) and \(\UB_i \gets 1\).
\FOR{each sentence \(y \in \mathcal{Y}\)}
    \FOR{each token \(y_s\) in \(y\) (with \(s\) as the token index)}
        \STATE \(i \gets s \mod n\).
        \STATE Compute \(k \gets \pi^{-1}(y_s)\).
        \STATE Compute the cumulative probabilities $C_{k-1}$ and $C_k$ with respect to permutation $\pi$, where $C_k = \sum_{j=1}^{k} p\bigl(\pi(j) \mid \text{prefix}\bigr)$
        and \(C_0 = 0\).
        \STATE \textbf{Update bounds:} $\LB_i \gets \max(\LB_i,\, C_{k-1}), \quad \UB_i \gets \min(\UB_i,\, C_k)$.
    \ENDFOR
\ENDFOR
\RETURN \(\{(\LB_i, \UB_i)\}_{i=1}^n\)
\end{algorithmic}
\end{algorithm}

With this principle, we use Algorithm \ref{alg:reverse_prng} to estimate the secret keys when the permutation is known. When the attacker only knows a partial ordering over a subset of tokens, we could modify Algorithm~\ref{alg:reverse_prng} with the same principle to estimate the secret keys. The details are in Algorithm \ref{alg:reverse_prng_partial} (Appendix~\ref{appendix:partial}).

\subsection{Unknown secret key sequence \(\{\xi_i\}\) and permutation \(\{\pi\}\)}
\label{sec:two-token-single-step}

In general, if both the random key sequence \(\{\xi_i\}\) and the permutation \(\pi\) are unknown, it is impossible to determine them uniquely. An interesting symmetry emerges in the parameterization of the watermarking process. Specifically, replacing each \(\xi_i\) with its complement \(1 - \xi_i\) and simultaneously reversing the permutation \(\pi\) leaves the watermarking process unchanged. Formally, given a permutation \(\pi\), define its reverse counterpart \(\pi'\) as $\pi' = \text{reverse}(\pi)$. 
Then the transformation $\xi_i \to 1 - \xi_i, \pi \to \pi'$ results in the same watermarking behavior almost surely. Formally, we show in \Cref{thm: iden_main_result} that this is the only possible alternative parametrization. The proof is in \Cref{appendix:math}.
\begin{theorem}
\label{thm: iden_main_result}
Given two secret keys $\xi, \hat{\xi} \in [0,1]$, two permutations $\pi, \hat{\pi}$ over the vocabulary and a probability distribution $p \in \Delta(\cV)$. Let $S(p, \xi, \pi)$ denote the token selection function that outputs the watermarked token in \Cref{alg:generate}. Suppose the selection function satisfies $S\bigl(p,\xi,\pi\bigr) = S\bigl(p,\widehat{\xi},\widehat{\pi}\bigr)$ for almost every $p\in\Delta(\mathcal V)$.
Then exactly one of the following must hold: either (i) $\pi=\widehat{\pi}$ and $\xi=\widehat{\xi}$ or (ii) $\pi=\mathrm{reverse}(\widehat{\pi})$ and $\xi=1-\widehat{\xi}$.  
%
\end{theorem}

We propose a method to learn \emph{one} of these two equivalent parameterizations.
 Our approach relies on constructing queries that force the LM to choose randomly between two candidate tokens, allowing us to record the relative order between any pair of tokens. With these pairwise ordering, a comparison-based sorting algorithm can recover the global ordering \(\pi\) (or its reverse) in \(O(|\cV| \log |\cV|)\) queries.

Specifically, we define a query interface, \texttt{QueryLLM(a, b)}, with two tokens \(a\) and \(b\) as input. With carefully designed prompts similar to~\citet{chen2024mark}, this interface ensures that the model considers only the two candidate tokens and assigns them equal probabilities. Under the ITS mechanism watermark, the model’s selection between \(a\) and \(b\) is governed only by the hidden permutation \(\pi\) and the corresponding secret key.

In our scheme, if the model outputs token \(b\), we interpret this as \(a < b\); otherwise, we interpret it as \(b < a\). With this ordering definition, we apply a comparison-based sorting algorithm (\eg Merge Sort) to recover an ordered sequence of tokens. The resulting order will correspond either to \(\pi\) or its reverse, which are equally useful. This result reduces the problem to the case where \(\pi\) is known. The algorithm details and analysis are summarized in Algorithm \ref{alg:recover_ordering} (Appendix~\ref{appendix:merge_sort}).

\section{Experimental Evaluation}
\label{sec:experiment}
To validate the effectiveness of our watermark spoofing methodology, we conduct a series of experiments under the most challenging threat model where the attacker has no prior knowledge of the watermark’s secret permutation  \(\pi\) or secret key \(\{\xi\}\) (section~\ref{sec:two-token-single-step}  and Appendix \ref{appendix:merge_sort}). After inferring the permuted vocabulary, the problem reduces to the case where the permutation is known (or partially known), but the secret key still needs to be estimated (section~\ref{sec:reverse_engineer_key} and Appendix \cref{appendix:partial}). Our evaluation addresses three key questions: \textbf{(Q1)} Can an attacker successfully carry out spoofing attacks on the \citet{kuditipudi2024robust}'s watermark? \textbf{(Q2)} Do the different LMs affect the spoofing result? and \textbf{(Q3)} How much computational resources does the attacker need to spoof successfully? 

We test our approach on four different LLMs: LLAMA-3.1-8B-Instruct~\citep{meta_llama_3_1_8b_instruct}, Mistral-7B-Instruct~\citep{jiang2023mistral7b}, Gemma-7b~\citep{gemmateam2024gemmaopenmodelsbased}, and OPT-125M~\citep{zhang2022opt}– using the OpenGen benchmark prompts~\citep{krishna2023paraphrasing}. We employ three evaluation metrics: the watermark detection $p$-value from \Cref{alg:detection}, cosine similarity using nomic embed models~\cite{nussbaum2025trainingsparsemixtureexperts} to calculate semantic similarity, and the output perplexity (PPL) to measure text fluency and quality. A successful spoofing attack should generate texts with $p$-value less than $\alpha=0.05$ (\Cref{appendix:prelim}), PPL on par with ordinary model outputs and cosine similarity score close to 1~\citep{7577578}. For additional experimental detail and results, see \Cref{appendix:experiment}.

\paragraph{(Q1) Successful spoofing attacks on \citet{kuditipudi2024robust}'s watermark.}
For the first research question, we answer positively: our proposed algorithms can successfully spoof the inverse-transform-sampling based watermark by \citet{kuditipudi2024robust}. 
First, we evaluate the spoofing results from the estimated permutation \(\pi\) and secret-key \(\xi_i\). We follow the procedure outlined in~\Cref{sec:methodology} to generate new text samples using the recovered permutation and secret key values (or subsets in the partial order scenarios). Our experimental setup compares three types of generated text: \textbf{(1)} genuine watermarked text produced by the watermarked LLM (using the true secret key and permutation) – this serves as a baseline for expected detector behavior and text quality, \textbf{(2)} non-watermarked text from the same model (without watermarking mechanism) – this serves as a control for typical LLM output, and \textbf{(3)} \textbf{spoofed text} produced by our attack (using the recovered permutation and secret key). We generate 100 samples in each category with 50 tokens per sample. We then apply the watermark detection and compute the perplexity of each sample using the base LlaMA-3.1-8B model to generate the perplexity scores. In addition, we use the Nomic Embed Model~\citep{nussbaum2025trainingsparsemixtureexperts} to generate embeddings to calculate the cosine similarity between spoofed text and the non-watermarked text as well as between the genuine watermarked text and the non-watermarked text. The results of this evaluation are summarized in \Cref{tab:median-results} and discussed below. Additional results are in \Cref{fig:pval_boxplots_threshold}, and~\Cref{appendix:experiment}. 

\Cref{tab:median-results} shows the median $p$-values for 100 samples of each generation type across all language models. Watermarked (WM) text yields low $p$-values ($p< 0.05$), and non-watermarked (Non-WM) text shows high values ($p> 0.05$), as expected.  Across all LLMs used in the experiment, with partial knowledge $\pi = 50\%$, majority of spoofed samples remain below the detection threshold, \ie the median $p$-values remain highly significant. 
We see some outliers with high $p$-values, especially in our larger models, bringing the mean $p$-value up, as seen in \Cref{tab:mean-results-threshold}. Across all models, detection starts to fail at $\pi = 25\%$ with many spoofed samples detected as non-watermarked. These results demonstrate that our method reliably fools the detector at moderate-to-high permutation recovery levels (\(50\%\)–\(100\%\)) across all tested models.

In addition to the detection test, we examine the quality of the spoofed text to ensure that our attack does not significantly degrade the coherence of the generated content. Concretely, we evaluate the text quality using two metrics: Perplexity (PPL) and Cosine Similarity. We compute perplexity using a separate reference language model (LLaMA 3.1 8B) that is different from the generation model. We decode each output using the generation model's tokenizer to preserve lexical fidelity, then re-tokenize the resulting text with the perplexity model's tokenizer for scoring. To mitigate distortion from rare token artifacts or malformed completions, we filter the samples with perplexity values that exceed the 95th percentile within each batch (similar to how~\citet{jovanovic2024watermarkstealinglargelanguage} define their attack success metric). For completion, we also report the results using unfiltered samples in \Cref{tab:mean-results-nothreshold} and \Cref{tab:median-results-nothreshold} in \Cref{appendix:experiment}. Cosine similarity measures how similar two vectors are regardless of their magnitude~\citep{Steck_2024}. We generate learned embeddings using the pretrained nomic embed language model ~\citep{nussbaum2025trainingsparsemixtureexperts} to turn each generated text into a vector which captures the meaning of the sentences. We then calculate the cosine similarity of the non-watermarked text vectors compared to the watermarked text vectors and spoofed text vectors which tells us how close the vectors point in the same directions. \Cref{tab:median-results} presents these statistics for the same sets of samples, where low perplexity and high cosine similarity indicates more fluent text. 

In \Cref{tab:median-results} and \Cref{fig:pval_boxplots_threshold}, we observe that the spoofed texts (Spoof 50 and Spoof 25) have almost the same perplexity distribution as the watermarked text (WM). Our spoofed texts remain fluent, especially when the watermark key recovery is mostly accurate. With a fully known permutation, the spoofed text’s perplexity is only slightly higher on average than genuine watermarked outputs across all LLMs. With these metrics as viable proxies for human preference~\citep{zheng2023judgingllmasajudgemtbenchchatbot, chiang2023largelanguagemodelsalternative}, we can confidently conclude that humans would be unlikely to notice any difference in quality, and the outputs are coherent and on-topic given the prompts. In \Cref{appendix:experiment}, we provide additional examples of texts generated by our spoofing attack, the watermarked output by the watermarked LLM and the base non-watermarked output using the same prompt. 
%


\begin{table}
\caption{Comparison of baseline and spoofed outputs with various LMs. We report the median $p$-value (p-val) for detection test, and a combination of perplexity (PPL) and cosine similarity (co-sim) for the text quality assessment.}
\label{tab:median-results}
\centering
\begin{tabular}{lccccc}
\toprule
\textbf{Model} & \multicolumn{2}{c}{\textbf{Non-WM}} & \multicolumn{3}{c}{\textbf{WM}} \\ \cmidrule(r){2-3} \cmidrule(r){4-6}
 & p-val & PPL & p-val & PPL & co-sim \\ \midrule
LLaMA-3.1-8B & 0.5 & 5.13 & 1.0e-04 & 16.32 & 0.871 \\
Mistral-7B & 0.4 & 8.91 & 1.0e-04 & 51.49 & 0.861 \\
Gemma-7B & 0.5 & 9.89 & 1.0e-04 & 59.02 & 0.85 \\
OPT-125M & 0.5 & 8.77 & 1.0e-04 & 132.46 & 0.836 \\
\bottomrule
\end{tabular}
\vspace{2mm}
\begin{tabular}{lcccccc}
\toprule
\textbf{Model}  & \multicolumn{3}{c}{\textbf{Spoof 50}} & \multicolumn{3}{c}{\textbf{Spoof 25}} \\ \cmidrule(r){2-4} \cmidrule(r){5-7} 
 & p-val & PPL & co-sim & p-val & PPL & co-sim \\ \midrule
LLaMA-3.1-8B & 1.0e-04 & 26.51 & 0.871 & 5.5e-04 & 24.09 & 0.859 \\
Mistral-7B & 1.0e-04 & 56.26 & 0.864 & 1.0e-04 & 68.48 & 0.863 \\
Gemma-7B & 1.0e-04 & 81.41 & 0.837 & 1.0e-04 & 86.35 & 0.829 \\
OPT-125M & 1.0e-04 & 130.5 & 0.834 & 1.0e-04 & 133.69 & 0.845 \\
\bottomrule
\end{tabular}
\end{table}

\paragraph{(Q2) Larger models impact spoofing quality.}
In our results, we observe that Gemma-7B underperforms relative to other models in both spoofing detection success and output perplexity. This behavior may be due to the LLM's more diffuse token distribution and less consistent top-k token ranking, which impair permutation recovery and increase alignment cost under the ITS watermark detection metric. In contrast, LLaMA-3.1-8B-Instruct shows strong spoofing performance due to its highly consistent token ranking, low-entropy output distributions, and robust response to prompt-based token comparisons. These traits make it easier to recover the watermark permutation and generate spoofed text that passes detection with low perplexity. We found that smaller model like OPT-125M achieve strong spoofing performance, with spoofed outputs consistently passing the watermark detection test. We attribute this to their flatter output distributions, which make the ITS watermarking more tolerant to approximation errors in the spoofing parameters. We do see these models produce text with higher perplexity, reflecting their limited language fluency and smaller effective vocabulary. While spoofing is statistically easier, the generated text lacks the quality of outputs from larger LLMs. 
\paragraph{(Q3) Query-efficient spoofing attack compared to prior work.}
An advantage of our framework is its sample efficiency. By directly querying the model with well-crafted comparisons and leveraging the theoretical structure of the ITS watermarking scheme, we require a low number of queries to steal the watermark. 
Although the exact number of queries depends on the vocabulary size and the model’s consistency, it is typically in order of only tens of thousands for a full reconstruction. All of our attacks were carried out with a feasible amount of computation on open-source models. This improved efficiency translates to real-world cost savings for the adversary. If the target model is accessed via a paid API, an attacker’s job requiring a million queries could be cost-prohibitive~\citep{sadasivan2025aigeneratedtextreliablydetected}. These results show that stealing a distortion-free watermark can be more cost-effective, \ie an expert adversary can compromise the watermark with limited resources. 

For comparison, \citet{jovanovic2024watermarkstealinglargelanguage} reported using 30k queries with 800k tokens to learn token distribution statistics for breaking a watermark, while our permutation recovery requires significantly fewer queries (<100) and tokens (<50). Our proposed approach translates to a significant reduction in attack cost and time.

Our results in \Cref{tab:median-results} and \Cref{fig:pval_boxplots_threshold} show that perfect recovery of the permutation is not needed to carry out a spoofing attack - a partially correct permutation can suffice. We verify this by evaluating the attacker’s success rate when only some of the true permutation is known. Even with only approximately half of the permutation recovered, detecting and spoofing the watermarked text with high success rate and minimal PPL penalty is possible across all model types. As permutation knowledge increases, the attack is more effective. 

Our results show that even with moderate permutation recovery (\eg 50\% of tokens), the secret key estimation error remains low and enables successful spoofing. Overall, we demonstrate that an attacker can recover both components of the distortion-free watermark. With a sufficiently accurate estimate of \(\pi\) and \(\{\xi\}\)– achieved with a feasible number of queries – the attacker is then able to generate spoofed text.

In summary, our experimental evaluation demonstrates that an attacker can (i) \textbf{reverse-engineer} the secret watermark parameters of a distortion-free watermarked LLM and (ii) use these to \textbf{generate a large quantity of spoofed text} that \textbf{fools the watermark detector} and remains high-quality. These findings directly challenge recent theoretical claims about the robustness of distortion-free LLM watermarks. In the next section, we conclude with a discussion of the implications of our work and potential directions for future research.

\begin{figure}[!t]
    \centering
    \includegraphics[width=1\linewidth]{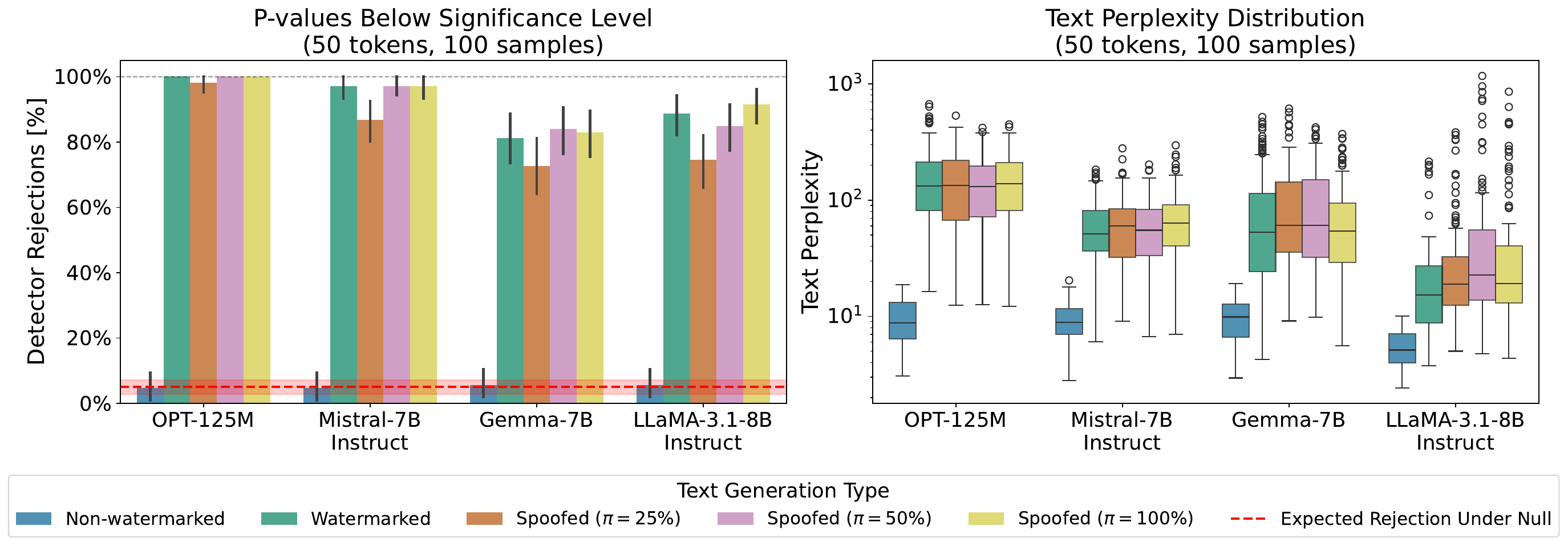}
    \caption{\textit{Left}: Percentage of $p$-values below the significance level ($\alpha=0.05$) for watermark detection in Algorithm~\ref{alg:detection}, across LLMs and known permutation $\pi$ proportions. Our spoofing attacks achieve similar rejection rates as watermarked samples from each model. \textit{Right:} Distribution of text perplexity for the generated samples, across LLMs and known permutation $\pi$ proportions. Our spoofed attacks do not exhibit significant difference in perplexity distributions with respect to the watermarked text.}
    \label{fig:pval_boxplots_threshold}
\end{figure}

\section{Conclusion and Future Work}
\label{sec:conclusion}
This paper presents a novel attack that breaks the distortion-free watermarking scheme for LLMs by using adaptive prompting and a sorting-based query strategy to retrieve the secret vocabulary permutation and key sequence. This allows us to generate `spoofed' high-quality text that we show to be statistically indistinguishable from genuine watermarked text. Our attack succeeds even when neither the watermark key nor the permutation is known ahead of time, with even a partial recovery of the permutation being sufficient. Our results challenge the security of distortion-free watermarks and complement existing watermark attacks, showing that the alternative ITS watermark scheme is also vulnerable and urging caution in using LLM watermarks as a line of defense. Future works include (i) developing watermarking schemes that are resilient to key recovery, such as dynamic key rotation or introducing randomness during generation, (ii) extending testing attacks under stricter constraints (\eg limited queries or partial detector access), (iii) scaling to larger models and (iv) exploring how spoofing interacts with other provenance mechanisms like model fingerprinting \cite{xu2019modeling}. 

\paragraph{Disclaimer.}
This paper was prepared for informational purposes by the Artificial Intelligence Research group of JPMorgan Chase \& Co. and its affiliates ("JP Morgan'') and is not a product of the Research Department of JP Morgan. JP Morgan makes no representation and warranty whatsoever and disclaims all liability, for the completeness, accuracy or reliability of the information contained herein. This document is not intended as investment research or investment advice, or a recommendation, offer or solicitation for the purchase or sale of any security, financial instrument, financial product or service, or to be used in any way for evaluating the merits of participating in any transaction, and shall not constitute a solicitation under any jurisdiction or to any person, if such solicitation under such jurisdiction or to such person would be unlawful.


\newpage
 \appendix
 \section{Identifiability of the reverse-Transform-Based Watermarking Scheme}
\label{appendix:math}
In this section, we establish an identifiability result for the reverse-transform-based watermarking scheme. Specifically, we address whether it is possible to uniquely recover both the secret key $\xi \in [0,1]$ and the permutation $\pi$, given outputs from queries to the language model. Formally, we define a selection function $S$ as follows: given a probability distribution $p \in \Delta(\mathcal{V})$, a secret key $\xi \in [0,1]$, and a permutation $\pi$, the output $S(p, \xi, \pi)$ denotes the token generated according to the procedure described in Algorithm~\ref{alg:generate}.

\begin{lemma}
\label{lem: iden_permutation}
Suppose the selection function satisfies 
\[
S(p,\xi,\pi) = S(p,\widehat{\xi},\widehat{\pi}),\quad \text\quad \text{for p almost everywhere on} \ \Delta(\mathcal{V}),
\]
where $\xi,\widehat{\xi}\in[0,1]$ are secret keys, and $\pi,\widehat{\pi}$ are permutations over the vocabulary.  
Then, it must hold that either $\pi = \widehat{\pi}$ or $\pi$ is the reverse permutation of $\widehat{\pi}$.
\end{lemma}

\begin{proof}

We prove by contradiction. Assume there exist two
permutations $\pi^{1}$ and $\pi^{2}$ that are neither identical nor mutual
reverses, such that 
$S(p,\xi,\pi^1)\equiv S(p,\widehat{\xi},\pi^2)$ for some secret keys $\xi$ and $\widehat{\xi}$.

Then by mathematical induction, we show that there exists tokens
$a,b,c$ such that \(b\) lies between \(a\) and \(c\) in~$\pi^{1}$ but not in
$\pi^{2}$.

Now consider a probability distribution $p(x)$ parametrized by $x$ over these tokens defined as follows:
\begin{itemize}
    \item Token $a$ is assigned probability $x$,
    \item Token $b$ is assigned probability $0.5(1 - x)$,
    \item Token $c$ is assigned probability $0.5(1 - x)$,
    \item All other tokens have probability zero.
\end{itemize}

As \(x\) increases from \(0\) to \(1\), the sampled tokens under permutation \(\pi^{1}\), $S(p(x),\xi,\pi^1)$ emerge in an ordered progression, sometimes involving only two of \(\{a,b,c\}\), sometimes all three, depending on the 
secret key $\xi$,
but in every case \(b\) invariably appears as the second token in that sequence.

However, under permutation $\pi^2$, since token $b$ is not positioned between tokens $a$ and $c$, the sequential appearance of these three tokens as $x$ changes will differ substantially. In fact, token $b$ will never appear as the second token shown in $S(p(x),\widehat{\xi},\pi^2)$, no matter what the $\widehat{\xi}$ is. This contradicts with the assumption $S(p,\xi,\pi^1)\equiv S(p,\widehat{\xi},\pi^2)$.
\end{proof}
\begin{lemma}
\label{lem: iden_secret_key}
Suppose the selection function satisfies 
\[
S(p,\xi,\pi) = S(p,\widehat{\xi},{\pi}),\quad \text{for p almost everywhere on} \ \Delta(\mathcal{V}),
\]
where $\xi,\widehat{\xi}\in[0,1]$ are secret keys. Then, it must hold that $\xi=\widehat{\xi}$.
\end{lemma}
\begin{proof}
Assume, for the sake of contradiction, that there exist keys
\(\xi \neq \widehat{\xi}\) such that
\(S(p,\xi,\pi)=S(p,\widehat{\xi},\pi)\) for a.e.\ \(p\in\Delta(\mathcal V)\).
Without loss of generality, let \(\xi > \widehat{\xi}\).

Denote the first and second tokens in the permutation \(\pi\) by
\(a := \pi(1)\) and \(b := \pi(2)\), respectively.
Consider a probability distribution $p(x)$
\[
p_{x}(a)=x,\qquad p_{x}(b)=1-x,\qquad p_{x}(c)=0\ \text{for all } c\notin\{a,b\},
\quad 0\le x\le 1.
\]

For any \(x\in(\widehat{\xi},\,\xi)\) we then have
\[
S\bigl(p_{x},\,\xi,\,\pi\bigr)=a
\quad\text{while}\quad
S\bigl(p_{x},\,\widehat{\xi},\,\pi\bigr)=b,
\]
contradicting the assumed equality of the two outputs. This finishes the proof.
\end{proof}

With \Cref{lem: iden_permutation} and \Cref{lem: iden_secret_key}, we are ready to present the proof of \Cref{thm: iden_main_result}, which highlights an intriguing equivalence between the two parameterizations:
%
%
%
%
\begin{proof}
By \cref{lem: iden_permutation}, we know either $\pi=\widehat{\pi}$ or $\pi=\mathrm{reverse}(\widehat{\pi})$ holds.
If $\pi=\widehat{\pi}$, we have $S(p,\xi,\pi)=S(p,\widehat{\xi},\pi)$ for almost every $p\in\Delta(\mathcal V)$.
Then, by \cref{lem: iden_secret_key}, we obtain $\xi=\widehat{\xi}$.

If $\pi=\mathrm{reverse}(\widehat{\pi})$, first note that
\begin{align}
\label{eq: reverse_fact}
S\bigl(p,\widehat{\xi},\widehat{\pi}\bigr)=
S\bigl(p,1-\widehat{\xi},\mathrm{inv}(\widehat{\pi})\bigr)
\quad\text{for almost every }p\in\Delta(\mathcal V),
\end{align}
which holds for any \(\widehat{\xi},\widehat{\pi}\).
Because \(\pi=\mathrm{inv}(\widehat{\pi})\), \eqref{eq: reverse_fact} becomes
\[
  S\bigl(p,\widehat{\xi},\widehat{\pi}\bigr)=
  S\bigl(p,1-\widehat{\xi},\pi\bigr)
  \quad\text{for almost every }p\in\Delta(\mathcal V).
\]
By assumption,
\[
  S\bigl(p,\xi,\pi\bigr)=S\bigl(p,\widehat{\xi},\widehat{\pi}\bigr)
  \quad\text{for almost every }p\in\Delta(\mathcal V),
\]
so we further have
\[
  S\bigl(p,\xi,\pi\bigr)=S\bigl(p,1-\widehat{\xi},\pi\bigr)
  \quad\text{for almost every }p\in\Delta(\mathcal V).
\]
Applying \cref{lem: iden_secret_key} again yields \(\xi=1-\widehat{\xi}\).

\end{proof}
 \newpage
 \section{Appendix: Additional Background on Watermarking}
\label{appendix:prelim}
%
In this section, we include the full algorithms used by~\citet{kuditipudi2024robust} to generate the watermark using inverse transform sampling (\Cref{alg:generate}) and to detect the watermark (\Cref{alg:detection}) for completion. 
\begin{algorithm}[ht]
\caption{Watermarked Text Generation via Inverse Transform Sampling}
\label{alg:generate}
\begin{algorithmic}[1]
\renewcommand{\algorithmicrequire}{\textbf{Input:}}
\renewcommand{\algorithmicensure}{\textbf{Output:}}

\REQUIRE Watermark key sequence $\xi = (\xi_1, \xi_2, \dots, \xi_n) \in \Xi^n$, generation length $m \in \mathbb{N}$, language model $p: \cV^* \to \Delta(\cV)$, bijective permutation function $\pi: [|\cV|] \to \cV$, vocabulary $\cV$.
\ENSURE Generated watermarked string $y = (y_1, \dots, y_m) \in \cV^m$

\FOR{$t = 1$ \TO $m$}
    \STATE $C_k = \sum_{i=1}^{k} p(\pi(i) | y_{:t-1})$, $\quad \forall k \in [|\cV|]\}$
    \STATE $\zeta = \xi_{t \bmod n}$
    \STATE $k = \min \{ k \mid C_k \geq \zeta \}$
    \STATE $y_t = \pi(k)$
\ENDFOR
\RETURN $y$
\end{algorithmic}
\end{algorithm}

At detection time, \Cref{alg:detection} is used to detect the watermark by using a permutation test to obtain a $p$- value. If the test returns a small $p$-value (typically chosen to be $p < 0.05$), then the text is watermarked.  
\begin{algorithm}[ht]
\caption{Watermark Detection via Alignment Cost}
\label{alg:detection}
\begin{algorithmic}[1]
\renewcommand{\algorithmicrequire}{\textbf{Input:}}
\renewcommand{\algorithmicensure}{\textbf{Output:}}
\REQUIRE 
\begin{itemize}
    \item Watermark key sequence $\xi = (\xi_1, \dots, \xi_n) \in [0,1]^n$.
    \item Candidate text $\tilde{y} = (\tilde{y}_1, \dots, \tilde{y}_m)$.
    \item Permutation $\pi: \{1,\dots,|\mathcal{V}|\} \rightarrow \mathcal{V}$.
    \item Language model $p: \cV^* \rightarrow \Delta(\cV)$.
    \item Number of random samples $T$ for the permutation test.
\end{itemize}
\ENSURE p-value $p$ indicating the likelihood that $\tilde{y}$ is watermarked.

\STATE $D \gets \textsc{ComputeAlignmentCost}(\tilde{y}, \xi)$
\STATE $c \gets 0$
\FOR{$i = 1$ \TO $T$}
    \STATE Sample $\mu^{(i)} \sim \mathcal{U}[0,1]^n$
    \STATE $D^{(i)} \gets \textsc{ComputeAlignmentCost}(\tilde{y}, \mu^{(i)})$
    \IF{$D \geq D^{(i)}$}
        \STATE $c \gets c + 1$
    \ENDIF
\ENDFOR
\STATE $p \gets \frac{1 + c}{T + 1}$
\RETURN $p$

\STATE \textbf{Function } \textsc{ComputeAlignmentCost}$(\tilde{y}, \xi)$:
    \STATE \quad $D \gets 0$
    \FOR{ $t = 1$ \TO $m$}
        \STATE \quad\quad $k_t \gets \pi^{-1}(\tilde{y}_t)$ \hfill $\rhd$ Obtain token index under $\pi$
        \STATE \quad\quad $C_{k_t} \gets \sum_{j=1}^{k_t} p\bigl(\pi(j) \mid \tilde{y}_{:t-1}\bigr)$ \hfill $\rhd$ Compute cumulative probability
        \STATE \quad\quad $\zeta_t \gets \xi_{t \bmod n}$ \hfill $\rhd$ Retrieve key element
        \STATE \quad\quad $d_t \gets \left| C_{k_t} - \zeta_t \right|$ \hfill $\rhd$ Token alignment cost
        \STATE \quad\quad $D \gets D + d_t$
    \ENDFOR
    \STATE \quad \textbf{return } $D$
\STATE \textbf{End Function}
\end{algorithmic}
\end{algorithm}

 \newpage
 \section{Appendix: Reverse-Engineering Pseudorandom Numbers from Watermarked Outputs with Partial Ordering}
\label{appendix:partial}

As mentioned in section~\ref{sec:reverse_engineer_key}, when the attacker only knows a partial ordering over a subset of tokens, we can use the same principles to estimate the secret keys. Algorithm~\ref{alg:reverse_prng_partial} below includes the details on how to modify Algorithm~\ref{alg:reverse_prng} in such settings.

\begin{algorithm}[!ht]
\caption{Reverse-Engineering Pseudorandom Numbers from Watermarked Outputs with Partial Ordering}
\label{alg:reverse_prng_partial}
\begin{algorithmic}[1]
\renewcommand{\algorithmicrequire}{\textbf{Input:}}
\renewcommand{\algorithmicensure}{\textbf{Output:}}
\REQUIRE 
\begin{itemize}
    \item A set of output sentences \(\mathcal{Y}\) generated by the watermarked model.
    \item A subset \(S \subset V\) with known ordering given by \(\pi^{-1}\) (i.e., for each \(t \in S\), \(\pi^{-1}(t)\) is known).
    \item Watermark key length \(n\) (i.e., there are \(n\) pseudorandom numbers \(\xi_1, \xi_2, \dots, \xi_n\)).
\end{itemize}
\ENSURE Estimated bounds for each pseudorandom number \(\xi_i\) (i.e., lower bound \(LB_i\) and upper bound \(UB_i\) for \(i=1,\dots,n\)).

\STATE \textbf{Initialization:} For \(i = 1\) to \(n\), set \(LB_i \gets 0\) and \(UB_i \gets 1\).
\FOR{each sentence \(y \in \mathcal{Y}\)}
    \FOR{each token \(y_s\) in \(y\) (with \(s\) as the token index)}
        \STATE \(i \gets s \mod n\)
        \IF{\(y_s \in S\)}
            \STATE Let \(r \gets \pi^{-1}(y_s)\).
            \STATE Compute the lower temporary bound:
            \[
            L_{\text{temp}} = \sum_{\substack{t \in S \\ \pi^{-1}(t) < r}} p\bigl(t \mid \text{prefix}\bigr).
            \]
            \STATE Compute the upper temporary bound:
            \[
            U_{\text{temp}} = 1 - \sum_{\substack{t \in S \\ \pi^{-1}(t) > r}} p\bigl(t \mid \text{prefix}\bigr).
            \]
            \STATE \textbf{Update bounds:}
            \[
            LB_i \gets \max\bigl(LB_i,\, L_{\text{temp}}\bigr), \quad UB_i \gets \min\bigl(UB_i,\, U_{\text{temp}}\bigr).
            \]
        \ENDIF
    \ENDFOR
\ENDFOR
\RETURN \(\{(LB_i, UB_i)\}_{i=1}^n\)
\end{algorithmic}
\end{algorithm}

 \newpage
 \section{Appendix: Recover Permutation via MergeSort}
\label{appendix:merge_sort}

In this section we include the details of the algorithm to recover the permutation order via MergeSort, as mentioned in section~\ref{sec:two-token-single-step}.

\begin{algorithm}[H]
\caption{Recover Global Ordering via QueryLLM and MergeSort}
\label{alg:recover_ordering}
\begin{algorithmic}[1]
\renewcommand{\algorithmicrequire}{\textbf{Input:}}
\renewcommand{\algorithmicensure}{\textbf{Output:}}
\REQUIRE Vocabulary of tokens \(T = \{t_1, t_2, \ldots, t_n\}\)
\ENSURE Ordered token sequence \(T_{\text{sorted}}\) equivalent to the hidden permutation \(\pi\) (or its reverse)

\STATE \(T_{\text{sorted}} \gets \text{MERGESORT}(T)\)
\RETURN \(T_{\text{sorted}}\)

\STATE \textbf{Function } \texttt{QueryLLM(\(a, b\))}:
\STATE \quad Query the language model with a prompt that forces a random choice between \(a\) and \(b\).
\STATE \quad \textbf{If} the model outputs token \(b\), \textbf{then} interpret as \(a < b\);
\STATE \quad \quad \textbf{Else} interpret as \(b < a\).
\STATE \quad \textbf{Return} the corresponding comparison result.
\STATE \textbf{End Function}
\STATE \textbf{Function } MERGESORT(\(arr\)):
\STATE \quad \textbf{If} length(\(arr\)) \(\leq 1\) 
\STATE \quad \quad \quad \textbf{return} \(arr\)
\STATE \quad \(mid \gets \lfloor \text{length}(arr)/2 \rfloor\)
\STATE \quad \(left \gets \text{MERGESORT}(arr[0:mid])\)
\STATE \quad \(right \gets \text{MERGESORT}(arr[mid:])\)
\STATE \quad \RETURN MERGE(\(left, right\))
\STATE \textbf{End Function}
\STATE \textbf{Function } MERGE(\(left, right\)):
\STATE \quad \(merged \gets\) empty list; \(i \gets 0\); \(j \gets 0\)
\STATE \quad \textbf{While} \(i < \text{length}(left)\) and \(j < \text{length}(right)\):
\STATE \quad\quad \textbf{If} QueryLLM(\(left[i],\, right[j]\)) returns "a < b":
\STATE \quad\quad\quad Append \(left[i]\) to \(merged\)
\STATE \quad\quad\quad \(i \gets i + 1\)
\STATE \quad\quad \textbf{Else}:
\STATE \quad\quad\quad Append \(right[j]\) to \(merged\)
\STATE \quad\quad\quad \(j \gets j + 1\)
\STATE \quad \textbf{End While}
\STATE \quad Append remaining elements of \(left[i:]\) to \(merged\)
\STATE \quad Append remaining elements of \(right[j:]\) to \(merged\)
\STATE \quad \RETURN \(merged\)
\STATE \textbf{End Function}
\end{algorithmic}
\end{algorithm}

We provide the following result to illustrate why Algorithm~\ref{alg:recover_ordering} is indeed a valid approach in our setting.

\begin{fact}
\label{fact_algorithm7}
Let the probability distribution \( p \) be defined over tokens as
\[
p(a) = 0.5, \quad p(b) = 0.5, \quad p(c) = 0 \text{ for all other tokens } c.
\]
Suppose the permutation \(\pi\) satisfies \(\pi^{-1}(b) > \pi^{-1}(a)\). Then the selection output satisfies:
\[
S(p, \xi, \pi) =
\begin{cases}
a, & \text{if } \xi \leq 0.5, \\
b, & \text{if } \xi > 0.5.
\end{cases}
\]
\end{fact}
The proof of \cref{fact_algorithm7} is straight forward, therefore we omit it.

With this fact established, the validity of \cref{alg:recover_ordering} becomes clear: if the corresponding secret key satisfies $\xi \leq 0.5$, then \texttt{QueryLLM}$(a,b)$ returns the token appearing earlier in the permutation $\pi$, causing the algorithm to recover $\textit{inv}(\pi)$. Otherwise, it returns the token appearing later, enabling direct recovery of $\pi$. In either case, the algorithm remains valid.

 \newpage
 \section{Appendix: Additional Experiments}
\label{appendix:experiment}

\subsection{Evaluating Estimated Vocabulary Permutation \(\pi\) and Estimated Secret Key}
First, we investigate whether an attacker can successfully recover the secret permutation \(\pi\) and key sequence \(\{\xi\}\). As described in Section~\ref{sec:two-token-single-step}, when the permutation is unknown, the attacker can estimate it by adaptively querying the LLM with pairwise token comparisons (the \texttt{QueryLLM(a,b)} procedure in \Cref{appendix:merge_sort}) and aggregating these comparisons via a merge-sort algorithm. This process would accurately estimate the true permutation in an ideal scenario with a perfectly consistent LLM that always follows instructions. In practice, LLMs may not provide a full-ordered preference for all token pairs (due to model uncertainties or equal probabilities). We address this issue by forcing a decision in ambiguous comparisons: if the model returns a tie, we break the tie arbitrarily to obtain a complete ordered list. We find this strategy effective: using well-engineered comparison prompts, an attacker can obtain an accurate estimation of the permutation over the vocabulary. We observe that smaller model (\eg OPT-125M) are often easier to spoof, \ie spoofed generations more consistently pass watermark detection despite only partial permutation recovery. We attribute this to two primary factors: (i) smaller models tend to have lower-entropy token distributions, making the ITS watermarking mechanism more brittle to approximate key values and permutation errors; and (ii) the lower diversity of output sequences in smaller models reduces alignment cost variance, leading to more forgiving p-values under the permutation test. However, while spoofing detection is easier, the output quality suffers: larger models like LLaMA-3.1-8B and Mistral-7B produce significantly more fluent and coherent text, as evidenced by consistently lower perplexity  and higher cosine similarity scores.

Figure~\ref{fig:key_mae_and_table} shows the mean absolute error (MAE) when estimating the secret keys as a function of the percentage of known permutation $\pi$ for the model vocabulary. As expected, we observe that the more information on $\pi$, the more accurate the estimation process, with the estimation error being virtually unchanged after at least $60\%$ of the permutation is known, validating the use of Algorithm~\ref{alg:reverse_prng} and the theoretical insights in Appendix~\ref{appendix:math}. Additionally, we note that in practice a successful spoofing attack might depend on the dataset, model and prompt used, and hence we leave for a future work exploring the connection between the estimation error of secret keys and the feasibility of spoofing attacks.

\begin{figure}[!ht]
    \centering
    \includegraphics[width=.6\linewidth]{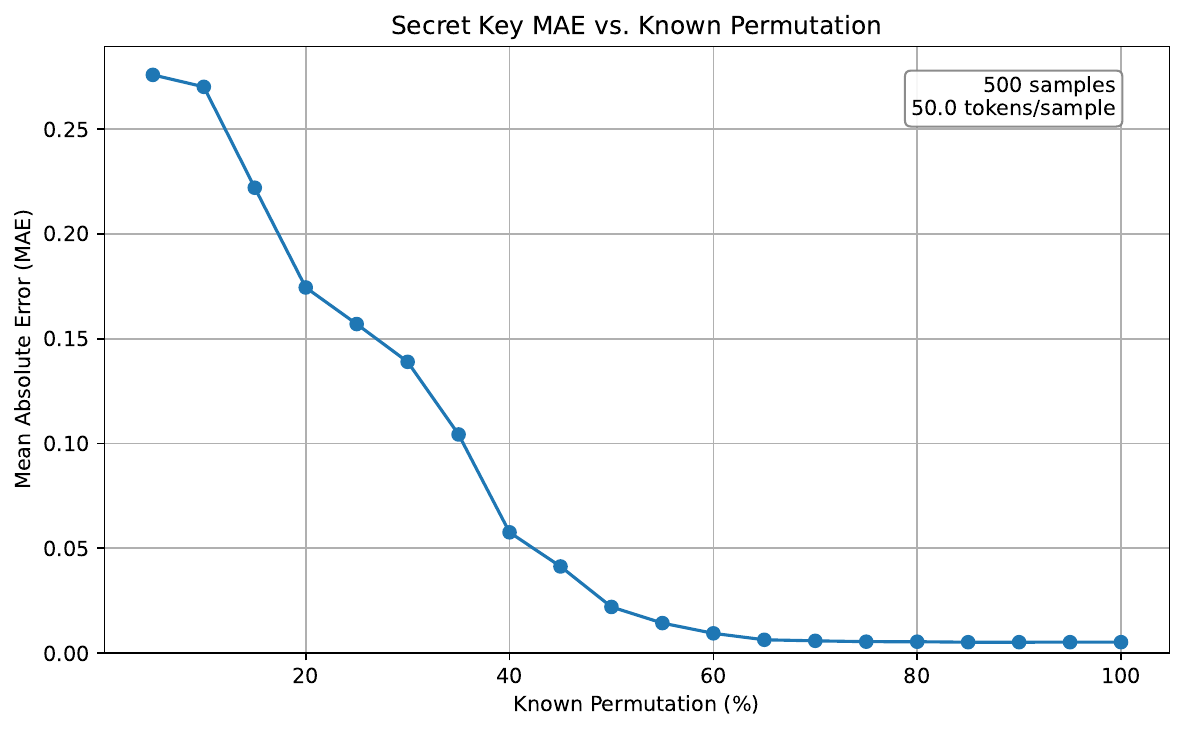}
    \caption{Secret key estimation error (MAE) as a function of known permutation fraction $\pi$ using 500 samples of 50 tokens each.}
    \label{fig:key_mae_and_table}
\end{figure}

\subsection{Additional spoofing evaluations}
In \Cref{sec:experiment}, we have presented \Cref{tab:median-results} containing the summary statistics for evaluating the spoofing attack on different LLMs. For completeness, we also provide additional tables of the same format in the appendix. The main differences between the results in \Cref{tab:median-results} and these additional tables are two-fold: (i) in \Cref{tab:median-results-nothreshold} and~\ref{tab:mean-results-nothreshold}, we report the results without filtering outliers as described in \Cref{sec:experiment} (we filter the samples with PPL values that exceed the 95th percentile within each batch); and (ii) in \Cref{tab:mean-results-threshold} and~\ref{tab:mean-results-nothreshold}, we report the mean $p$-value across $100$ samples instead of the median $p$-value. Furthermore, we provide the comparison of $p$-value and perplexity across different LLMs without filtering outliers in PPL values \Cref{fig:bar_box_perp_pval_plots_no_thresh}. In these new results, our observations are consistent with the findings in \Cref{sec:experiment}. Our proposed approach can successfully spoof the watermark by \citet{kuditipudi2024robust} across different LLMs, and the spoofing results generally improves with more capable LLMs.  

\begin{table}[ht!]
\caption{Comparison of baseline and spoofed outputs with various LLMs. We report the mean $p$-value (p-val) for detection test, and a combination of perplexity (PPL) and cosine similarity (co-sim) with thresholding for the text quality assessment.}
\label{tab:mean-results-threshold}
\centering
\begin{tabular}{lccccc}
\toprule
\textbf{Model} & \multicolumn{2}{c}{\textbf{Non-WM}} & \multicolumn{3}{c}{\textbf{WM}} \\ \cmidrule(r){2-3} \cmidrule(r){4-6}
 & p-val & PPL & p-val & PPL & co-sim \\ \midrule
LLaMA-3.1-8B & 0.5 & 5.67 & 3.4e-02 & 29.79 & 0.866 \\
Mistral-7B & 0.5 & 9.56 & 8.2e-03 & 65.08 & 0.855 \\
Gemma-7B & 0.5 & 9.96 & 8.7e-02 & 119.41 & 0.846 \\
OPT-125M & 0.5 & 9.84 & 1.1e-04 & 176.73 & 0.834 \\
\bottomrule
\end{tabular}
\vspace{2mm}
\begin{tabular}{lcccccc}
\toprule
\textbf{Model} & \multicolumn{3}{c}{\textbf{Spoof 50}} & \multicolumn{3}{c}{\textbf{Spoof 25}} \\ \cmidrule(r){2-4} \cmidrule(r){5-7}
 & p-val & PPL & co-sim & p-val & PPL & co-sim \\ \midrule
LLaMA-3.1-8B & 5.6e-02 & 108.04 & 0.869 & 1.2e-01 & 55.53 & 0.854 \\
Mistral-7B & 6.1e-03 & 65.83 & 0.86 & 3.4e-02 & 74.79 & 0.858 \\
Gemma-7B & 7.6e-02 & 119.84 & 0.838 & 1.6e-01 & 134.15 & 0.833 \\
OPT-125M & 1.0e-04 & 150.08 & 0.832 & 1.4e-02 & 155.6 & 0.836 \\
\bottomrule
\end{tabular}
\end{table}

\begin{table}[ht!]
\caption{Comparison of baseline and spoofed outputs with various LLMs. We report the median $p$-value (p-val) for detection test, and a combination of perplexity (PPL) and cosine similarity (co-sim) without thresholding for the text quality assessment.}
\label{tab:median-results-nothreshold}
\centering
\begin{tabular}{lccccc}
\toprule
\textbf{Model} & \multicolumn{2}{c}{\textbf{Non-WM}} & \multicolumn{3}{c}{\textbf{WM}} \\ \cmidrule(r){2-3} \cmidrule(r){4-6}
 & p-val & PPL & p-val & PPL & co-sim \\ \midrule
LLaMA-3.1-8B & 0.5 & 5.5 & 1.0e-04 & 18.09 & 0.871 \\
Mistral-7B & 0.5 & 13.01 & 1.0e-04 & 66.91 & 0.861 \\
Gemma-7B & 0.5 & 11.22 & 1.0e-04 & 65.53 & 0.85 \\
OPT-125M & 0.5 & 10.08 & 1.0e-04 & 142.79 & 0.836 \\
\bottomrule
\end{tabular}
\vspace{2mm}
\begin{tabular}{lcccccc}
\toprule
\textbf{Model} & \multicolumn{3}{c}{\textbf{Spoof 50}} & \multicolumn{3}{c}{\textbf{Spoof 25}} \\ \cmidrule(r){2-4} \cmidrule(r){5-7}
 & p-val & PPL & co-sim & p-val & PPL & co-sim \\ \midrule
LLaMA-3.1-8B & 1.0e-04 & 24.52 & 0.871 & 1.0e-04 & 27.47 & 0.859 \\
Mistral-7B & 1.0e-04 & 66.89 & 0.864 & 1.0e-04 & 80.61 & 0.863 \\
Gemma-7B & 1.0e-04 & 84.61 & 0.837 & 1.0e-04 & 92.22 & 0.829 \\
OPT-125M & 1.0e-04 & 142.04 & 0.834 & 1.0e-04 & 147.35 & 0.845 \\
\bottomrule
\end{tabular}
\end{table}

\begin{table}[ht!]
\caption{Comparison of baseline and spoofed outputs with various LLMs. We report the mean $p$-value (p-val) for detection test, and a combination of perplexity (PPL) and cosine similarity (co-sim) without thresholding for the text quality assessment.}
\label{tab:mean-results-nothreshold}
\centering
\begin{tabular}{lccccc}
\toprule
\textbf{Model} & \multicolumn{2}{c}{\textbf{Non-WM}} & \multicolumn{3}{c}{\textbf{WM}} \\ \cmidrule(r){2-3} \cmidrule(r){4-6}
 & p-val & PPL & p-val & PPL & co-sim \\ \midrule
LLaMA-3.1-8B & 0.5 & 6.27 & 3.2e-02 & 409.25 & 0.866 \\
Mistral-7B & 0.5 & 14.69 & 7.8e-03 & 89.11 & 0.855 \\
Gemma-7B & 0.5 & 12.04 & 8.1e-02 & 172.82 & 0.846 \\
OPT-125M & 0.5 & 11.62 & 1.1e-04 & 199.52 & 0.834 \\
\bottomrule
\end{tabular}
\vspace{2mm}
\begin{tabular}{lcccccc}
\toprule
\textbf{Model} & \multicolumn{3}{c}{\textbf{Spoof 50}} & \multicolumn{3}{c}{\textbf{Spoof 25}} \\ \cmidrule(r){2-4} \cmidrule(r){5-7}
 & p-val & PPL & co-sim & p-val & PPL & co-sim \\ \midrule
LLaMA-3.1-8B & 3.3e-02 & 299.5 & 0.869 & 1.6e-02 & 717.58 & 0.854 \\
Mistral-7B & 5.7e-03 & 87.97 & 0.86 & 4.2e-02 & 98.18 & 0.858 \\
Gemma-7B & 1.0e-01 & 153.86 & 0.838 & 1.5e-01 & 163.26 & 0.833 \\
OPT-125M & 1.0e-04 & 185.79 & 0.832 & 1.3e-02 & 172.88 & 0.836 \\
\bottomrule
\end{tabular}
\end{table}

\begin{figure}[ht!]
    \centering
    \includegraphics[width=1.0\linewidth]{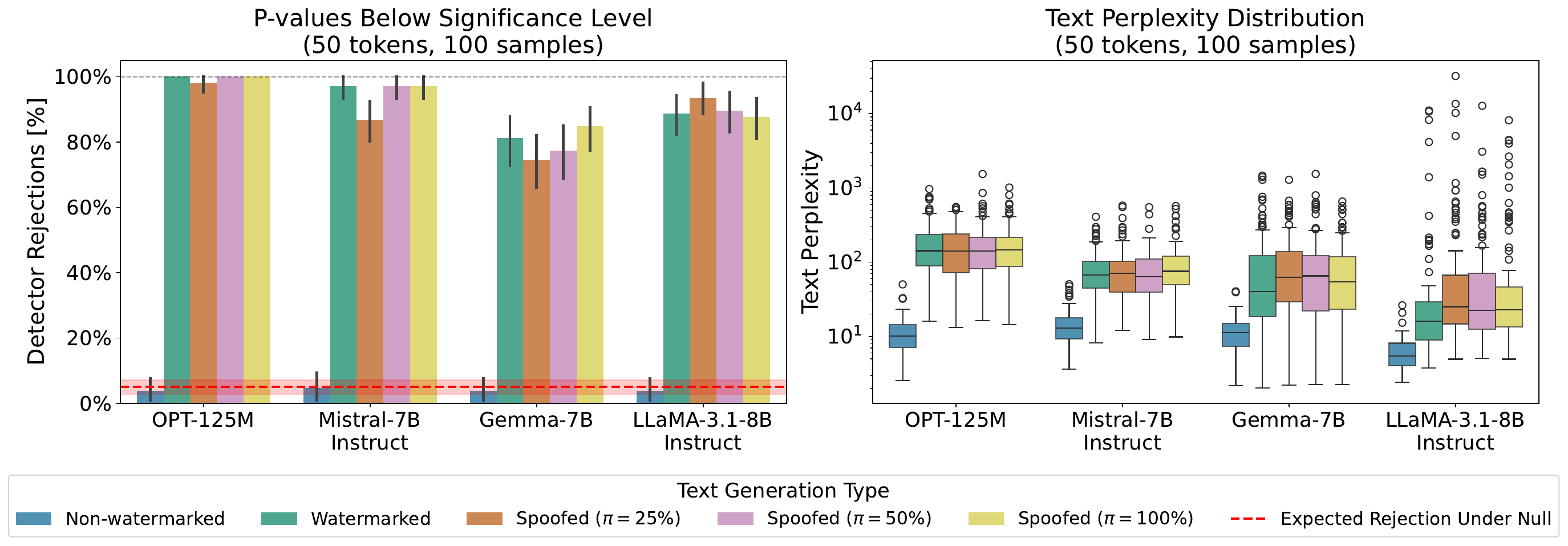}
    \caption{Spoofing results as in Figure~\ref{fig:pval_boxplots_threshold}, but without removing generated low-quality text. \textit{Left}: Percentage of $p$-values below the significance level ($\alpha=0.05$) for watermark detection in Algorithm~\ref{alg:detection}, across LLMs and known permutation $\pi$ proportions. Our spoofing attacks achieve similar rejection rates as watermarked samples from each model. \textit{Right:} Distribution of text perplexity for the generated samples, across LLMs and known permutation $\pi$ proportions. Our spoofed attacks do not exhibit significant difference in perplexity distributions with respect to the watermarked text.}
    \label{fig:bar_box_perp_pval_plots_no_thresh}
\end{figure}


\newpage

\subsection{Spoofing Examples}
We provide additional examples of our spoofing attacks. Specifically, under the same prompt "What is Quantum Computing", we report the outputs from a non-watermarked LM, watermarked LM and our spoofing attack as well as the $p$-value (for watermark detection) and perplexity score (for text quality) using LLaMA-3.1-8B-Instruct for all generation types.

\textbf{Prompt: \texttt{"What is Quantum Computing?"} (200 tokens)}
\paragraph{Watermarked Output}
\begin{quote}
A quantum computer is a developing field of technology that uses the principles of quantum mechanics, such as superposition and entanglement, to perform calculations and operations that are beyond the capabilities of traditional computers.

Quantum computers are designed to solve complex problems, such as simulating chemical reactions or factoring large numbers, exponentially faster than classical computers. They work by manipulating qubits (quantum bits) that can exist in multiple states simultaneously, allowing for massive parallel processing and exploration of vast solution spaces.

The core concept of quantum computing is to exploit the inherent properties of quantum systems to perform operations that cannot be simulated with classical computers. These properties include:

\textbf{Quantum superposition:} A qubit can represent both 0 and 1 simultaneously, enabling it to process multiple possibilities at once.

\textbf{Quantum entanglement:} Qubits can become “entangled,” meaning their properties are connected, allowing for instant communication and correlation between them.

\textbf{Quantum measurement:} Measuring a quantum system collapses it into a definite state, which is used to extract information from the system.
\end{quote}

\paragraph{Non-Watermarked Output}
\begin{quote}
What is quantum computing? Quantum computing is a new paradigm for computing that uses the principles of quantum mechanics to perform calculations and operations on data.

Quantum computing is based on the principles of superposition, entanglement, and interference, which allow for the processing of vast amounts of data in parallel. Quantum computing has the potential to solve complex problems that are currently unsolvable with classical computers.

Quantum computing is a new and rapidly evolving field that has the potential to revolutionize the way we approach complex problems in fields such as cryptography, optimization, and machine learning. Quantum computing is based on the principles of quantum mechanics, which describe the behavior of matter and energy at the smallest scales. Quantum computing uses quantum bits, or qubits, which can exist in multiple states simultaneously, unlike classical bits, allowing for massive parallelism in computation.
\end{quote}







\paragraph{Spoofed Output}
\begin{quote}
A quantum computer is a computer that can use quantum mechanical phenomena, such as superposition and entanglement, that enable it to solve certain problems that are difficult for a classical computer to solve, faster and more efficiently. Quantum computing is a new and rapidly evolving field which has seen significant advances in recent years.

Quantum computers are based on quantum mechanical systems, such as quantum bits (qubits) and quantum gates, and are able to perform certain operations that are beyond the capabilities of classical computers. Quantum computers have the potential to revolutionize the field of computation and may have important applications in fields such as cryptography, optimization, and simulation.

A classical computer works on binary (0, 1) bits to process information and make decisions. In contrast, a quantum computer will work on quantum bits (0, 1, 2, and other states), known as qubits, which allow for complex operations and parallelism.

Quantum computing has the potential to accelerate countless applications and industries.
\end{quote}

\begin{table}[ht!]
\centering
\caption{Summary statistics comparing generation types by detection p-value and perplexity.}
\label{tab:summary-stats}
\begin{tabular}{@{}lcc@{}}
\toprule
\textbf{Generation Type} & \textbf{p-value} & \textbf{Perplexity} \\
\midrule
Watermarked       & $9.999 \times 10^{-5}$ & 2.74 \\
Non-Watermarked   & 0.469                  & 1.56 \\
Spoofed           & $9.999 \times 10^{-5}$ & 11.49 \\
\bottomrule
\end{tabular}
\end{table}

 \newpage
 \section{Computational Resources}
\label{appendix:computer_resources}

Both experiments and evaluations are run on an Ubuntu Machine with 96 Intel Xeon Platinum 8259CL CPUs, 384GB of RAM, a storage volume of 500 Gb and 8 A100 NVIDIA Tensor Core GPUs.





\begin{thebibliography}{34}
\providecommand{\natexlab}[1]{#1}
\providecommand{\url}[1]{\texttt{#1}}
\expandafter\ifx\csname urlstyle\endcsname\relax
  \providecommand{\doi}[1]{doi: #1}\else
  \providecommand{\doi}{doi: \begingroup \urlstyle{rm}\Url}\fi

\bibitem[Aaronson(2023)]{aaronson2023openai}
S.~Aaronson.
\newblock `{R}eform' {AI} {A}lignment with {S}cott {A}aronson.
\newblock AXRP - the AI X-risk Research Podcast, 2023.
\newblock URL
  \url{https://axrp.net/episode/2023/04/11/episode-20-reform-ai-alignment-scott-aaronson.html}.

\bibitem[AI(2024)]{meta_llama_3_1_8b_instruct}
M.~AI.
\newblock Llama 3.1 8b instruct, 2024.
\newblock URL \url{https://huggingface.co/meta-llama/Llama-3.1-8B-Instruct}.

\bibitem[Chen et~al.(2024)Chen, Wu, Guo, and Huang]{chen2024mark}
R.~Chen, Y.~Wu, J.~Guo, and H.~Huang.
\newblock De-mark: Watermark removal in large language models.
\newblock \emph{arXiv preprint arXiv:2410.13808}, 2024.

\bibitem[Cheng et~al.(2025)Cheng, Guo, Li, and Sigal]{cheng2025revealing}
Y.~Cheng, H.~Guo, Y.~Li, and L.~Sigal.
\newblock Revealing weaknesses in text watermarking through self-information
  rewrite attacks.
\newblock \emph{arXiv preprint arXiv:2505.05190}, 2025.

\bibitem[Chiang and yi~Lee(2023)]{chiang2023largelanguagemodelsalternative}
C.-H. Chiang and H.~yi~Lee.
\newblock Can large language models be an alternative to human evaluations?,
  2023.
\newblock URL \url{https://arxiv.org/abs/2305.01937}.

\bibitem[Gloaguen et~al.(2024{\natexlab{a}})Gloaguen, Jovanovi{\'c}, Staab, and
  Vechev]{gloaguen2024black}
T.~Gloaguen, N.~Jovanovi{\'c}, R.~Staab, and M.~Vechev.
\newblock Black-box detection of language model watermarks.
\newblock \emph{arXiv preprint arXiv:2405.20777}, 2024{\natexlab{a}}.

\bibitem[Gloaguen et~al.(2024{\natexlab{b}})Gloaguen, Jovanovi{\'c}, Staab, and
  Vechev]{gloaguen2024discovering}
T.~Gloaguen, N.~Jovanovi{\'c}, R.~Staab, and M.~Vechev.
\newblock Discovering clues of spoofed lm watermarks.
\newblock \emph{arXiv preprint arXiv:2410.02693}, 2024{\natexlab{b}}.

\bibitem[Grattafiori et~al.(2024)Grattafiori, Dubey, Jauhri, Pandey, Kadian,
  Al-Dahle, Letman, Mathur, Schelten, Vaughan,
  et~al.]{grattafiori2024llama3herdmodels}
A.~Grattafiori, A.~Dubey, A.~Jauhri, A.~Pandey, A.~Kadian, A.~Al-Dahle,
  A.~Letman, A.~Mathur, A.~Schelten, A.~Vaughan, et~al.
\newblock The llama 3 herd of models, 2024.

\bibitem[Gu et~al.(2024)Gu, Li, Liang, and
  Hashimoto]{gu2024learnabilitywatermarkslanguagemodels}
C.~Gu, X.~L. Li, P.~Liang, and T.~Hashimoto.
\newblock On the learnability of watermarks for language models, 2024.
\newblock URL \url{https://arxiv.org/abs/2312.04469}.

\bibitem[Jiang et~al.(2023)Jiang, Sablayrolles, Mensch, Bamford, Chaplot,
  de~las Casas, Bressand, Lengyel, Lample, Saulnier, Lavaud, Lachaux, Stock,
  Scao, Lavril, Wang, Lacroix, and Sayed]{jiang2023mistral7b}
A.~Q. Jiang, A.~Sablayrolles, A.~Mensch, C.~Bamford, D.~S. Chaplot, D.~de~las
  Casas, F.~Bressand, G.~Lengyel, G.~Lample, L.~Saulnier, L.~R. Lavaud, M.-A.
  Lachaux, P.~Stock, T.~L. Scao, T.~Lavril, T.~Wang, T.~Lacroix, and W.~E.
  Sayed.
\newblock Mistral 7b, 2023.
\newblock URL \url{https://arxiv.org/abs/2310.06825}.

\bibitem[Jovanović et~al.(2024)Jovanović, Staab, and
  Vechev]{jovanovic2024watermarkstealinglargelanguage}
N.~Jovanović, R.~Staab, and M.~Vechev.
\newblock Watermark stealing in large language models, 2024.
\newblock URL \url{https://arxiv.org/abs/2402.19361}.

\bibitem[Kirchenbauer et~al.(2023)Kirchenbauer, Geiping, Wen, Katz, Miers, and
  Goldstein]{kirchenbauer2023watermark}
J.~Kirchenbauer, J.~Geiping, Y.~Wen, J.~Katz, I.~Miers, and T.~Goldstein.
\newblock A watermark for large language models.
\newblock \emph{arXiv preprint arXiv:2301.10226}, 2023.

\bibitem[Krishna et~al.(2023)Krishna, Song, Karpinska, Wieting, and
  Iyyer]{krishna2023paraphrasing}
K.~Krishna, Y.~Song, M.~Karpinska, J.~F. Wieting, and M.~Iyyer.
\newblock Paraphrasing evades detectors of ai-generated text, but retrieval is
  an effective defense.
\newblock In \emph{Thirty-seventh Conference on Neural Information Processing
  Systems (NeurIPS)}, 2023.
\newblock URL \url{https://openreview.net/forum?id=WbFhFvjjKj}.

\bibitem[Kuditipudi et~al.(2024)Kuditipudi, Thickstun, Hashimoto, and
  Liang]{kuditipudi2024robust}
R.~Kuditipudi, J.~Thickstun, T.~Hashimoto, and P.~Liang.
\newblock Robust distortion-free watermarks for language models.
\newblock \emph{Transactions on Machine Learning Research}, 2024.
\newblock ISSN 2835-8856.
\newblock URL \url{https://openreview.net/forum?id=FpaCL1MO2C}.

\bibitem[Lahitani et~al.(2016)Lahitani, Permanasari, and Setiawan]{7577578}
A.~R. Lahitani, A.~E. Permanasari, and N.~A. Setiawan.
\newblock Cosine similarity to determine similarity measure: Study case in
  online essay assessment.
\newblock In \emph{2016 4th International Conference on Cyber and IT Service
  Management}, pages 1--6, 2016.
\newblock \doi{10.1109/CITSM.2016.7577578}.

\bibitem[Liu et~al.(2024)Liu, Pan, Lu, Li, Hu, Zhang, Wen, King, Xiong, and
  Yu]{liu2024survey}
A.~Liu, L.~Pan, Y.~Lu, J.~Li, X.~Hu, X.~Zhang, L.~Wen, I.~King, H.~Xiong, and
  P.~Yu.
\newblock A survey of text watermarking in the era of large language models.
\newblock \emph{ACM Computing Surveys}, 57\penalty0 (2):\penalty0 1--36, 2024.

\bibitem[Naseh et~al.(2023)Naseh, Krishna, Iyyer, and Houmansadr]{Naseh_2023}
A.~Naseh, K.~Krishna, M.~Iyyer, and A.~Houmansadr.
\newblock Stealing the decoding algorithms of language models.
\newblock In \emph{Proceedings of the 2023 ACM SIGSAC Conference on Computer
  and Communications Security}, CCS ’23, page 1835–1849. ACM, Nov. 2023.
\newblock \doi{10.1145/3576915.3616652}.
\newblock URL \url{http://dx.doi.org/10.1145/3576915.3616652}.

\bibitem[Ning et~al.(2024)Ning, Chen, Zhong, Zhang, Wang, Li, Zhang, Zhang, and
  Zheng]{ning2024mcgmark}
K.~Ning, J.~Chen, Q.~Zhong, T.~Zhang, Y.~Wang, W.~Li, Y.~Zhang, W.~Zhang, and
  Z.~Zheng.
\newblock Mcgmark: An encodable and robust online watermark for llm-generated
  malicious code.
\newblock \emph{arXiv preprint arXiv:2408.01354}, 2024.

\bibitem[Nussbaum and
  Duderstadt(2025)]{nussbaum2025trainingsparsemixtureexperts}
Z.~Nussbaum and B.~Duderstadt.
\newblock Training sparse mixture of experts text embedding models, 2025.
\newblock URL \url{https://arxiv.org/abs/2502.07972}.

\bibitem[OpenAI(2022)]{chatgpt}
OpenAI.
\newblock Chatgpt: Optimizing language models for dialogue, 2022.
\newblock URL \url{https://openai.com/blog/chatgpt}.

\bibitem[Ouyang et~al.(2022)Ouyang, Wu, Jiang, Almeida, Wainwright, Mishkin,
  Zhang, Agarwal, Slama, Ray, Schulman, Hilton, Kelton, Miller, Simens, Askell,
  Welinder, Christiano, Leike, and
  Lowe]{ouyang2022traininglanguagemodelsfollow}
L.~Ouyang, J.~Wu, X.~Jiang, D.~Almeida, C.~L. Wainwright, P.~Mishkin, C.~Zhang,
  S.~Agarwal, K.~Slama, A.~Ray, J.~Schulman, J.~Hilton, F.~Kelton, L.~Miller,
  M.~Simens, A.~Askell, P.~Welinder, P.~Christiano, J.~Leike, and R.~Lowe.
\newblock Training language models to follow instructions with human feedback,
  2022.
\newblock URL \url{https://arxiv.org/abs/2203.02155}.

\bibitem[Pang et~al.(2024)Pang, Hu, Zheng, and Smith]{pang2024no}
Q.~Pang, S.~Hu, W.~Zheng, and V.~Smith.
\newblock No free lunch in llm watermarking: Trade-offs in watermarking design
  choices.
\newblock In \emph{The Thirty-eighth Annual Conference on Neural Information
  Processing Systems}, 2024.

\bibitem[Piet et~al.(2024)Piet, Sitawarin, Fang, Mu, and
  Wagner]{piet2024markwordsanalyzingevaluating}
J.~Piet, C.~Sitawarin, V.~Fang, N.~Mu, and D.~Wagner.
\newblock Mark my words: Analyzing and evaluating language model watermarks,
  2024.
\newblock URL \url{https://arxiv.org/abs/2312.00273}.

\bibitem[Sadasivan et~al.(2025)Sadasivan, Kumar, Balasubramanian, Wang, and
  Feizi]{sadasivan2025aigeneratedtextreliablydetected}
V.~S. Sadasivan, A.~Kumar, S.~Balasubramanian, W.~Wang, and S.~Feizi.
\newblock Can ai-generated text be reliably detected?, 2025.
\newblock URL \url{https://arxiv.org/abs/2303.11156}.

\bibitem[Steck et~al.(2024)Steck, Ekanadham, and Kallus]{Steck_2024}
H.~Steck, C.~Ekanadham, and N.~Kallus.
\newblock Is cosine-similarity of embeddings really about similarity?
\newblock In \emph{Companion Proceedings of the ACM Web Conference 2024}, WWW
  ’24, page 887–890. ACM, May 2024.
\newblock \doi{10.1145/3589335.3651526}.
\newblock URL \url{http://dx.doi.org/10.1145/3589335.3651526}.

\bibitem[Team et~al.(2024)Team, Mesnard, Hardin, Dadashi, Bhupatiraju, Pathak,
  Sifre, Rivière, Kale, Love, Tafti, Hussenot, Sessa, Chowdhery, Roberts,
  Barua, Botev, Castro-Ros, Slone, Héliou, Tacchetti, Bulanova, Paterson,
  Tsai, Shahriari, Lan, Choquette-Choo, Crepy, Cer, Ippolito, Reid,
  Buchatskaya, Ni, Noland, Yan, Tucker, Muraru, Rozhdestvenskiy, Michalewski,
  Tenney, Grishchenko, Austin, Keeling, Labanowski, Lespiau, Stanway, Brennan,
  Chen, Ferret, Chiu, Mao-Jones, Lee, Yu, Millican, Sjoesund, Lee, Dixon, Reid,
  Mikuła, Wirth, Sharman, Chinaev, Thain, Bachem, Chang, Wahltinez, Bailey,
  Michel, Yotov, Chaabouni, Comanescu, Jana, Anil, McIlroy, Liu, Mullins,
  Smith, Borgeaud, Girgin, Douglas, Pandya, Shakeri, De, Klimenko, Hennigan,
  Feinberg, Stokowiec, hui Chen, Ahmed, Gong, Warkentin, Peran, Giang, Farabet,
  Vinyals, Dean, Kavukcuoglu, Hassabis, Ghahramani, Eck, Barral, Pereira,
  Collins, Joulin, Fiedel, Senter, Andreev, and
  Kenealy]{gemmateam2024gemmaopenmodelsbased}
G.~Team, T.~Mesnard, C.~Hardin, R.~Dadashi, S.~Bhupatiraju, S.~Pathak,
  L.~Sifre, M.~Rivière, M.~S. Kale, J.~Love, P.~Tafti, L.~Hussenot, P.~G.
  Sessa, A.~Chowdhery, A.~Roberts, A.~Barua, A.~Botev, A.~Castro-Ros, A.~Slone,
  A.~Héliou, A.~Tacchetti, A.~Bulanova, A.~Paterson, B.~Tsai, B.~Shahriari,
  C.~L. Lan, C.~A. Choquette-Choo, C.~Crepy, D.~Cer, D.~Ippolito, D.~Reid,
  E.~Buchatskaya, E.~Ni, E.~Noland, G.~Yan, G.~Tucker, G.-C. Muraru,
  G.~Rozhdestvenskiy, H.~Michalewski, I.~Tenney, I.~Grishchenko, J.~Austin,
  J.~Keeling, J.~Labanowski, J.-B. Lespiau, J.~Stanway, J.~Brennan, J.~Chen,
  J.~Ferret, J.~Chiu, J.~Mao-Jones, K.~Lee, K.~Yu, K.~Millican, L.~L. Sjoesund,
  L.~Lee, L.~Dixon, M.~Reid, M.~Mikuła, M.~Wirth, M.~Sharman, N.~Chinaev,
  N.~Thain, O.~Bachem, O.~Chang, O.~Wahltinez, P.~Bailey, P.~Michel, P.~Yotov,
  R.~Chaabouni, R.~Comanescu, R.~Jana, R.~Anil, R.~McIlroy, R.~Liu, R.~Mullins,
  S.~L. Smith, S.~Borgeaud, S.~Girgin, S.~Douglas, S.~Pandya, S.~Shakeri,
  S.~De, T.~Klimenko, T.~Hennigan, V.~Feinberg, W.~Stokowiec, Y.~hui Chen,
  Z.~Ahmed, Z.~Gong, T.~Warkentin, L.~Peran, M.~Giang, C.~Farabet, O.~Vinyals,
  J.~Dean, K.~Kavukcuoglu, D.~Hassabis, Z.~Ghahramani, D.~Eck, J.~Barral,
  F.~Pereira, E.~Collins, A.~Joulin, N.~Fiedel, E.~Senter, A.~Andreev, and
  K.~Kenealy.
\newblock Gemma: Open models based on gemini research and technology, 2024.
\newblock URL \url{https://arxiv.org/abs/2403.08295}.

\bibitem[Wu and Chandrasekaran(2024)]{wu2024bypassingllmwatermarkscoloraware}
Q.~Wu and V.~Chandrasekaran.
\newblock Bypassing llm watermarks with color-aware substitutions, 2024.
\newblock URL \url{https://arxiv.org/abs/2403.14719}.

\bibitem[Xu et~al.(2019)Xu, Skoularidou, Cuesta-Infante, and
  Veeramachaneni]{xu2019modeling}
L.~Xu, M.~Skoularidou, A.~Cuesta-Infante, and K.~Veeramachaneni.
\newblock Modeling tabular data using conditional gan.
\newblock \emph{Advances in neural information processing systems}, 32, 2019.

\bibitem[Zhang et~al.(2024{\natexlab{a}})Zhang, Edelman, Francati, Venturi,
  Ateniese, and Barak]{zhang2024watermarkssandimpossibilitystrong}
H.~Zhang, B.~L. Edelman, D.~Francati, D.~Venturi, G.~Ateniese, and B.~Barak.
\newblock Watermarks in the sand: Impossibility of strong watermarking for
  generative models, 2024{\natexlab{a}}.
\newblock URL \url{https://arxiv.org/abs/2311.04378}.

\bibitem[Zhang et~al.(2024{\natexlab{b}})Zhang, Hussain, Neekhara, and
  Koushanfar]{zhang2024remark}
R.~Zhang, S.~S. Hussain, P.~Neekhara, and F.~Koushanfar.
\newblock $\{$REMARK-LLM$\}$: A robust and efficient watermarking framework for
  generative large language models.
\newblock In \emph{33rd USENIX Security Symposium (USENIX Security 24)}, pages
  1813--1830, 2024{\natexlab{b}}.

\bibitem[Zhang et~al.(2022)Zhang, Roller, Goyal, Artetxe, Chen, Chen, Dewan,
  Diab, Li, Lin, et~al.]{zhang2022opt}
S.~Zhang, S.~Roller, N.~Goyal, M.~Artetxe, M.~Chen, S.~Chen, C.~Dewan, M.~Diab,
  X.~Li, X.~V. Lin, et~al.
\newblock Opt: Open pre-trained transformer language models.
\newblock \emph{arXiv preprint arXiv:2205.01068}, 2022.

\bibitem[Zhang et~al.(2024{\natexlab{c}})Zhang, Zhang, Zhang, Zhang, Chen, Hu,
  Gill, and Pan]{zhang2024largelanguagemodelwatermark}
Z.~Zhang, X.~Zhang, Y.~Zhang, L.~Y. Zhang, C.~Chen, S.~Hu, A.~Gill, and S.~Pan.
\newblock Large language model watermark stealing with mixed integer
  programming, 2024{\natexlab{c}}.
\newblock URL \url{https://arxiv.org/abs/2405.19677}.

\bibitem[Zhao et~al.(2023)Zhao, Ananth, Li, and
  Wang]{zhao2023provablerobustwatermarkingaigenerated}
X.~Zhao, P.~Ananth, L.~Li, and Y.-X. Wang.
\newblock Provable robust watermarking for ai-generated text, 2023.
\newblock URL \url{https://arxiv.org/abs/2306.17439}.

\bibitem[Zheng et~al.(2023)Zheng, Chiang, Sheng, Zhuang, Wu, Zhuang, Lin, Li,
  Li, Xing, Zhang, Gonzalez, and
  Stoica]{zheng2023judgingllmasajudgemtbenchchatbot}
L.~Zheng, W.-L. Chiang, Y.~Sheng, S.~Zhuang, Z.~Wu, Y.~Zhuang, Z.~Lin, Z.~Li,
  D.~Li, E.~P. Xing, H.~Zhang, J.~E. Gonzalez, and I.~Stoica.
\newblock Judging llm-as-a-judge with mt-bench and chatbot arena, 2023.
\newblock URL \url{https://arxiv.org/abs/2306.05685}.

\end{thebibliography}
\end{document}